\documentclass[10pt,a4paper]{article}
\usepackage[utf8]{inputenc}
\usepackage[T1]{fontenc}
\usepackage[english]{babel}
\usepackage{hyperref}
\usepackage{subfigure}
\usepackage{amsmath,amsthm,enumerate}
\usepackage{amssymb}
\usepackage{xspace}
\usepackage{tikz}
\usetikzlibrary{calc}
\usepackage{shapepar}
\usepackage{comment}
\usepackage[color=green!30]{todonotes}
\usetikzlibrary{fit}
\usetikzlibrary{shapes.geometric}
\usepackage{lmodern}
\usetikzlibrary{decorations.pathmorphing}
\tikzset{snake it/.style={decorate, decoration=snake}}
\usepackage[color=green!30]{todonotes}	%
\usepackage[hmargin=2.5cm,vmargin=3cm]{geometry}

\newtheorem{theorem}{Theorem}[section]
\newtheorem{proposition}[theorem]{Proposition}
\newtheorem{lemma}[theorem]{Lemma}
\newtheorem{observation}[theorem]{Observation}

\newtheorem{definition}[theorem]{Definition}

\newtheorem{question}[theorem]{Question}
\newtheorem{conjecture}[theorem]{Conjecture}
\newtheorem{claim}[theorem]{Claim}
\newtheorem{remark}[theorem]{Remark}

\newcommand{\smallqed}{{\tiny ($\Box$)}}

\newcommand{\PBCOL}[1]{\textsc{#1-Colouring}\xspace}
\newcommand{\PBlistCOL}[1]{\textsc{#1-List-Colouring}\xspace}

\newcommand{\PBtropCOL}[1]{\textsc{#1-Tropical-Colouring}\xspace}

\newlength{\atextwidth}
\newcommand{\problemdec}[3]{
  \vspace{1mm}
\noindent\fbox{
  \begin{minipage}{\atextwidth}
  \begin{tabular*}{\textwidth}{@{\extracolsep{\fill}}lr} #1 \\ \end{tabular*}
  {\bf{Input:}} #2  \\
  {\bf{Question:}} #3
  \end{minipage}
  }
  \vspace{1mm}
}

\newcommand{\shortpaper}[1]{}
\newcommand{\longpaper}[1]{#1}

\begin{document}

\title{The complexity of tropical graph homomorphisms}
\author{Florent Foucaud\footnote{\noindent LIMOS, Universit\'e Blaise Pascal, Clermont-Ferrand (France). florent.foucaud@gmail.com}
\and Ararat Harutyunyan\footnote{\noindent IMT, Universit\'{e} Toulouse III (Paul Sabatier), Toulouse (France). aharutyu@math.univ-toulouse.fr}
\and Pavol Hell\footnote{\noindent School of Computing Science, Simon Fraser University, Burnaby (Canada). pavol@sfu.ca}
\and Sylvain Legay\footnote{\noindent LRI, Universit\'e Paris-Sud, Orsay (France). \{legay,yannis\}@lri.fr}
\and Yannis Manoussakis\footnotemark[4]
\and Reza Naserasr\footnote{\noindent CNRS - IRIF, Université Paris Diderot, Paris (France). E-mail: reza@irif.fr}
}

\maketitle

\vspace{-5mm}
\noindent\fbox{
  \begin{minipage}{\textwidth}
    {\bf \emph{Note to readers.} A shorter version of this article appeared in \href{http://dx.doi.org/10.1016/j.dam.2017.04.027}{\emph{Discrete Applied Mathematics} 229 (2017)}. The present extended version contains all missing proofs and additional figures.}
  \end{minipage}
}

\begin{abstract}
A tropical graph $(H,c)$ consists of a graph $H$ and a (not necessarily proper) vertex-colouring $c$ of $H$. Given two tropical graphs $(G,c_1)$ and $(H,c)$, a homomorphism of $(G,c_1)$ to $(H,c)$ is a standard graph homomorphism of $G$ to $H$ that also preserves the vertex-colours. We initiate the study of the computational complexity of tropical graph homomorphism problems. We consider two settings. First, when the tropical graph $(H,c)$ is fixed; this is a problem called \PBCOL{$(H,c)$}. Second, when the colouring of $H$ is part of the input; the associated decision problem is called \PBtropCOL{$H$}. Each \PBCOL{$(H,c)$} problem is a constraint satisfaction problem (CSP), and we show that a complexity dichotomy for the class of \PBCOL{$(H,c)$} problems holds if and only if the Feder--Vardi Dichotomy Conjecture for CSPs is true. This implies that \PBCOL{$(H,c)$} problems form a rich class of decision problems. On the other hand, we were successful in classifying the complexity of at least certain classes of \PBtropCOL{$H$} problems.
\end{abstract}

\section{Introduction}
Unless stated otherwise, the graphs considered in this paper are simple, loopless and finite. A \emph{homomorphism} $h$ of a graph $G$ to a graph $H$ is a mapping $h:V(G)\to V(H)$ such that adjacency is preserved by $h$, that is, the images of two adjacent vertices of $G$ must be adjacent in $H$. If such a mapping exists, we note $G\to H$. For a fixed graph $H$, given an input graph $G$, the decision problem \PBCOL{$H$} (whose name is derived from the proximity of the problem to proper vertex-colouring) consists of determining whether $G\to H$ holds. Problems of the form \PBCOL{$H$} for some fixed graph $H$, are called \emph{homomorphism problems}. A classic theorem of Hell and Ne\v{s}et\v{r}il~\cite{HN90} states a \emph{dichotomy} for this problem: if $H$ is bipartite, \PBCOL{$H$} is polynomial-time solvable; otherwise, it is NP-complete.

\medskip

\noindent\textbf{Tropical graphs.} As an extension of graph homomorphisms, homomorphisms of edge-coloured graphs have been studied, see for example~\cite{AM98,Bthesis,B94,BFHN15,BH00}. In this paper, we consider the variant where the \emph{vertices} are coloured. We initiate the study of \emph{tropical graph homomorphism problems}, in which the vertex sets of the graphs are partitioned into colour classes. Formally, a \emph{tropical graph} $(G,c)$ is a graph $G$ together with a (not necessarily proper) vertex-colouring $c:V(G)\to C$ of $G$, where $C$ is a set of colours. If $|C|=k$, we say that $(G,c)$ is a \emph{$k$-tropical graph}. Given two tropical graphs $(G,c_1)$ and $(H,c_2)$ (where the colour set of $c_1$ is a subset of the colour set of $c_2$), a homomorphism $h$ of $(G,c_1)$ to $(H,c_2)$ is a homomorphism of $G$ to $H$ that also preserves the colours, that is, for each vertex $v$ of $G$, $c_1(v)=c_2(h(v))$. For a fixed tropical graph $(H,c)$, problem \PBCOL{$(H,c)$} asks whether, given an input tropical graph $(G,c_1)$, we have $(G,c_1)\to (H,c)$.

\medskip

\noindent\textbf{The homomorphism factoring problem.} Brewster and MacGillivray defined the following related problem in~\cite{BM97}. For two fixed graphs $H$ and $Y$ and a homomorphism $h$ of $H$ to $Y$, the \textsc{$(H,h,Y)$-Factoring} problem takes as an input, a graph $G$ together with a homomorphism $g$ of $G$ to $Y$, and asks for the existence of a homomorphism $f$ of $G$ to $H$ such that $f=h\circ g$.
The \PBCOL{$(H,c)$} problem corresponds to \textsc{$(H,c,K_{|C|}^+)$-Factoring} where $K_{|C|}^+$ is the complete graph on $|C|$ vertices with all loops (and with $C$ the set of colours used by $c$). (Note that in~\cite{BM97}, loops were not considered.)

\medskip

\noindent\textbf{Constraint satisfaction problems (CSPs).} Graph homomorphism problems fall into a more general class of decision problems, the \emph{constraint satisfaction problems}, defined for \emph{relational structures}. A relational structure $S$ over a \emph{vocabulary} (a vocabulary is a set of pairs $(R_i,a_i)$ of relation names and arities) consists of a \emph{domain} $V(S)$ of vertices together with a set of relations corresponding to the vocabulary, that is, $R_i\subseteq V(S)^{a_i}$ for each relation $R_i$ of the vocabulary. Given two relational structures $S$ and $T$ over the same vocabulary, a homomorphism of $S$ to $T$ is a mapping $h:V(S)\to V(T)$ such that each relation $R_i$ is preserved, that is, for each subset of $V(S)^{a_i}$ of $R_i$ in $S$, its image set in $T$ also belongs to $R_i$. For a fixed relational structure $T$, \textsc{$T$-CSP} is the decision problem asking whether a given input relational structure has a homomorphism to $T$.

Using this terminology, a graph $H$ is a relational structure over the vocabulary $\{(A,2)\}$ consisting of a single binary relation $A$ (adjacency). Hence, \PBCOL{$H$} is a CSP. Further, \PBCOL{$(H,c)$} is equivalent to the problem \textsc{$C(H,c)$-CSP}, where $C(H,c)$ is obtained from $H$ by adding a set of $k$ unary relations to $H$ (one for each colour class of the $k$-colouring $c$).

\medskip

\noindent\textbf{The Dichotomy Conjecture.} In their celebrated paper~\cite{FV98}, Feder and Vardi posed the following conjecture.

\begin{conjecture}[Feder and Vardi~\cite{FV98}]\label{conj:CSP-dicho}
For every fixed relational structure $T$, \textsc{$T$-CSP} is polynomial-time solvable or NP-complete.
\end{conjecture}

Conjecture~\ref{conj:CSP-dicho} became known as the \emph{Dichotomy Conjecture} and has given rise to extensive work in this area, see for example~\cite{B03,B02,FH98,FH98csp,FHH99,FHH02}. If the conjecture holds, it would imply a fundamental distinction between CSP and the whole class NP. Indeed, the latter is known (unless P$=$NP) to contain so-called \emph{NP-intermediate} problems that are neither NP-complete nor polynomial-time solvable~\cite{L75}.

The Dichotomy Conjecture 
was motivated by several earlier dichotomy theorems for special cases, such as the one of Schaefer for binary structures~\cite{S78} or the one of Hell and Ne\v{s}et\v{r}il for undirected graphs, stated as follows.

\begin{theorem}[Hell and Ne\v{s}et\v{r}il Dichotomy \cite{HN90}]\label{thm:HN}
Let $H$ be an undirected graph. If $H$ is bipartite, then \PBCOL{$H$} is polynomial-time solvable. Otherwise, \PBCOL{$H$} is NP-complete.
\end{theorem}

\medskip

\noindent\textbf{Digraph homomorphisms.} Digraph homomorphisms are also well-studied in the context of complexity dichotomies. We will relate them to tropical graph homomorphisms. For a digraph $D$, \PBCOL{$D$} asks whether an input digraph admits a homomorphism to $D$, that is, a homomorphism of the underlying undirected graphs that also preserves the orientation of the arcs.

 While in the case of undirected graphs, the \PBCOL{H} problem is only polynomial time for graphs whose core is either $K_1$ or $K_2$, in the case of digraphs the problem remains polynomial time for a large class of digraphs which are cores. The classification of such cores has been one of the difficulties of the conjecture. Such classifications are given for certain interesting subclasses, see for example~\cite{BH90,BHM88,BHM95,BKN09,F01}. A proof of a conjectured classification of the general case has been announced while this paper was under review (see~\cite{FKR17}). If valid, this would imply the truth of the Dichotomy Conjecture, as Feder and Vardi~\cite{FV98} showed the following (seemingly weaker) statement to be equivalent to it.

\begin{conjecture}[Equivalent form of the Dichotomy Conjecture, Feder and Vardi~\cite{FV98}]\label{conj:digraph-dicho}
For every bipartite digraph $D$, \PBCOL{$D$} is polynomial-time solvable or NP-complete.
\end{conjecture}

In Section~\ref{sec:dicho-CSP}, similarly to its above reformulation (Conjecture~\ref{conj:digraph-dicho}),  we will show that the Dichotomy Conjecture has an equivalent formulation as a dichotomy for tropical homomorphisms problems. More precisely, we will show that the Dichotomy Conjecture is true if and only if its restriction to \PBCOL{$(H,c)$} problems, where $(H,c)$ is a $2$-tropical bipartite graph, also holds. In other words, one can say that the class of $2$-tropical bipartite graph homomorphisms is as rich as the whole class of CSPs.

For many digraphs $D$ it is known such that \PBCOL{$D$} is NP-complete. Such a digraph of order~$4$ and size~$5$ is presented in the book by Hell and Ne\v{s}et\v{r}il~\cite[page 151]{HNbook}. Such oriented trees are also known, see~\cite{HNZ96} or~\cite[page 158]{HNbook}; the smallest such known tree has order~$45$. A full dichotomy is known for oriented cycles~\cite{F01}; the smallest such NP-complete oriented cycle has order between~$24$ and~$36$~\cite{D16,F01}. Using these results, one can easily exhibit some NP-complete \PBCOL{$(H,c)$} problems. To this end, given a digraph $D$, we construct the $3$-tropical graph $T(D)$ as follows. Start with the set of vertices $V(D)$ and colour its vertices Blue. For each arc $\overrightarrow{uv}$ in $D$, add a path $ux_ux_vv$ of length~$3$ from $u$ to $v$ in $T(D)$, where $x_u$ and $x_v$ are two new vertices coloured Red and Green, respectively. The following fact is not difficult to observe.

\begin{proposition}\label{prop:digraph-reduc}
For any two digraphs $D_1$ and $D_2$, we have $D_1\to D_2$ if and only if $T(D_1)\to T(D_2)$.
\end{proposition}

By the above results on NP-complete \PBCOL{$D$} problems and Proposition~\ref{prop:digraph-reduc}, we obtain a $3$-tropical graph of order~$14$, a $3$-tropical tree of order~$133$, and a $3$-tropical cycle of order between~$72$ and~$108$ whose associated homomorphism problems are NP-complete. Nevertheless, in this paper, we exhibit (by using other reduction techniques) much smaller tropical graphs, trees and cycles $(H,c)$ with \PBCOL{$(H,c)$} NP-complete.

\medskip

\noindent\textbf{List homomorphisms.} Dichotomy theorems have also been obtained for a list-based extension of the class of homomorphism problems, the \emph{list-homomorphism problems}. In this setting, introduced by Feder and Hell in~\cite{FH98}, the input consists of a pair $(G,L)$, where $G$ is a graph and $L:V(G)\to 2^{V(H)}$ is a list assignment representing a set of allowed images for each vertex of $G$. For a fixed graph $H$, the decision problem \PBlistCOL{$H$} asks whether there is a homomorphism $h$ of $G$ to $H$ such that for each vertex $v$ of $G$, $h(v)\in L(v)$. Problem \PBlistCOL{$H$} can be seen as a generalization of \PBCOL{$H$}. Indeed, restricting \PBlistCOL{$H$} to the class of inputs where for each vertex $v$ of $G$, $L(v)=V(H)$, corresponds precisely to \PBCOL{$H$}. Therefore, if \PBCOL{$H$} is NP-complete, so is \PBlistCOL{$H$}. For this set of problems, a full complexity dichotomy has been established in a series of three papers~\cite{FH98,FHH99,FHH02}. We state the dichotomy result for simple graphs from~\cite{FHH99}, that is related to our work. (A circular arc graphs is an intersection graph of arcs on a cycle.)

\begin{theorem}[Feder, Hell and Huang~\cite{FHH99}]\label{thm:listhom-table}
If $H$ is a bipartite graph such that its complement is a circular arc graph, then \PBlistCOL{$H$} is polynomial-time solvable. Otherwise, \PBlistCOL{$H$} is NP-complete.
\end{theorem}


Given a tropical graph $(H,c)$, the problem \PBCOL{$(H,c)$} is equivalent to the restriction of \PBlistCOL{$H$} to instances $(G,L)$ where each list is the set of vertices  in one of the colour classes of $c$. Next, we introduce a less restricted variant of \PBlistCOL{$H$} that is also based on tropical graph homomorphisms.

\medskip

\noindent\textbf{The \PBtropCOL{\boldmath{$H$}} problem.} 
Given a fixed graph $H$, we introduce the decision problem \PBtropCOL{$H$}, whose instances consist of (1) a vertex-colouring $c$ of $H$ and (2) a tropical graph $(G,c_2)$. Then, \PBtropCOL{$H$} consists of deciding whether $(G,c_1)\to (H,c)$.

Alternatively, \PBtropCOL{$H$} is an instance restriction of \PBlistCOL{$H$} to instances with \emph{laminar lists}, that is, lists such that for each pair of distinct vertices $v_1, v_2 \in V(G)$, $L(v_1) = L(v_2)$ or $L(v_1) \cap L(v_2) =\emptyset$. (We remark that \PBtropCOL{$H$}, as well as \PBlistCOL{$H$}, can also be formulated as a CSP, where certain unary relations encode the list constraints: so-called \emph{full CSPs}, see~\cite{FH98csp} for details.)

Given the difficulty of studying \PBCOL{$(H,c)$} problems, as will be demonstrated in Section~\ref{sec:dicho-CSP}, the study of \PBtropCOL{$H$} problems will be the focus of the other parts of this paper. This study is directed by the following question.

\begin{question}\label{mainquestion}
For a given graph $H$, what is the complexity of \PBtropCOL{$H$}?
\end{question}

Clearly, \PBCOL{$(H,c)$} where each vertex receives the same colour, is computationally equivalent to \PBCOL{$H$}. Therefore, by the Hell-Ne\v{s}et\v{r}il dichotomy of Theorem~\ref{thm:HN}, if $H$ is non-bipartite, \PBtropCOL{$H$} is NP-complete. Furthermore, by the above formulation of \PBtropCOL{$H$} as an instance restriction of \PBlistCOL{$H$}, whenever \PBlistCOL{$H$} is polynomial-time solvable, so is \PBtropCOL{$H$}.


Thus, according to Theorems~\ref{thm:HN} and~\ref{thm:listhom-table}, all problems \PBtropCOL{$H$} where $H$ is not bipartite are NP-complete, and all problems \PBtropCOL{$H$} where $H$ is bipartite and its complement is a circular-arc graph are polynomial-time solvable. Thus, it remains to study \PBtropCOL{$H$} when $H$ belongs to the class of bipartite graphs whose complement is not a circular-arc graph. This class of graphs has been well-studied, and characterized by forbidden induced subgraphs~\cite{TM76}. It is a rich class of graphs that includes all cycles of length at least~$6$, all trees with at least one vertex from which there are three branches of length at least~$3$, and an many other graphs~\cite{TM76}.

Observe that for any induced subgraph $H'$ of a graph $H$, one can reduce \PBtropCOL{$H'$} to \PBtropCOL{$H$} by assigning, in the input colouring of $H$, a dummy colour to all the vertices of $H-H'$. Hence, if \PBtropCOL{$H$} is polynomial-time solvable, then \PBtropCOL{$H'$} is also polynomial-time solvable. Conversely, if \PBtropCOL{$H'$} is NP-complete, so is \PBtropCOL{$H$}. Therefore, to answer Question~\ref{mainquestion}, it is enough to consider minimal graphs $H$ such that \PBtropCOL{$H$} is NP-complete.

A first question is to study the case of minimal graphs $H$ for which \PBlistCOL{$H$} is NP-complete; such a list is known and it follows from Theorem~\ref{thm:listhom-table}. In particular, it contains all even cycles of length at least~$6$. In Section~\ref{sec:list_hom}, we show that for every even cycle $C_{2k}$ of length at least~$48$, \PBtropCOL{$C_{2k}$} is NP-complete. On the other hand,for every even cycle $C_{2k}$ of length at most~$12$, \PBtropCOL{$C_{2k}$} is polynomial-time solvable. Unfortunately, for each graph $H$ in the above-mentioned list that is not a cycle, \PBtropCOL{$H$} is polynomial-time solvable, and thus larger graphs will be needed in the quest of a similar characterization of NP-complete \PBtropCOL{$H$} problems.

In Section~\ref{sec:smallgraphs}, we show that for every bipartite graph $H$ of order at most~$8$, \PBtropCOL{$H$} is polynomial-time solvable, but there is a bipartite graph $H_9$ of order~$9$ such that \PBtropCOL{$H_9$} is NP-complete.

Finally, in Section~\ref{sec:trees}, we study the case of trees. We prove that for every tree $T$ of order at most~$11$, \PBtropCOL{$T$} is polynomial-time solvable, but there is a tree $T_{23}$ of order~$23$ such that \PBtropCOL{$T_{23}$} is NP-complete.

We remark that our NP-completeness results are finer than those that can be obtained from Proposition~\ref{prop:digraph-reduc}, in the sense that the orders of the obtained target graphs are much smaller. Similarly, we note that the results in~\cite{BM97} imply the existence of NP-complete \PBtropCOL{$H$} problems, and $H$ can be chosen to be a tree or a cycle. However, similarly as in Proposition~\ref{prop:digraph-reduc}, these results are also based on reductions from NP-complete \PBCOL{$D$} problems, where $H$ is obtained from the digraph $D$ by replacing each arc by a path (its length depends on $D$, but it is always at least~$3$). Thus, the NP-complete tropical targets obtained in~\cite{BM97} are trees of order at least~$133$ and cycles of order at least~$72$, which is much more than the ones exhibited in the present paper.

\section{Preliminaries and tools}\label{sec:prelim}

In this section we gather some necessary preliminary definitions and results.

\subsection{Isomorphisms, cores}

For tropical graph homomorphisms, we have the same basic notions and properties as in the theory of graph homomorphisms. A homomorphism of tropical graph $(G,c_1)$ to $(H,c_2)$ is an \emph{isomorphism} if it is a bijection and it acts bijectively on the set of edges.

\begin{definition}
The \emph{core} of a tropical graph $(G,c)$ is the smallest (in terms of the order) induced tropical subgraph $(G',c_{|G'})$ admitting a homomorphism of $(G,c)$ to $(G',c_{|G'})$.
\end{definition}

In the same way as for simple graphs, it can be proved that the core of a tropical graph is unique. A tropical graph $(G,c)$ is called a \emph{core} if its core is isomorphic to $(G,c)$ itself.
Moreover, we can restrict ourselves to studying only cores. Indeed it is not difficult to check that $(G,c_1)$ admits a homomorphism to $(H,c_2)$ if and only if the core of $(G,c_1)$ admits a homomorphism to the core of $(H,c_2)$.


\subsection{Formal definitions of the used computational problems}

We now formally define all the decision problems used in this paper.

\problemdec{\PBCOL{$H$}}{A (di)graph $G$.}{Does there exist a homomorphism of $G$  to $H$?}

\smallskip


\problemdec{\PBlistCOL{$H$}}{A graph $G$ and a list function $L:V(G)\to 2^{V(H)}$.}{Is there a homomorphism $f$ of $G$ to $H$ such that for every vertex $x$ of $G$, $f(x)\in L(x)$?}

\smallskip


\problemdec{\PBCOL{$(H,c)$}}{A tropical graph $(G,c_1)$.}{Does $(G,c_1)$ admit a homomorphism to $(H,c)$?}

\smallskip

\problemdec{\PBtropCOL{$H$}}{A vertex-colouring $c$ of $H$, and a tropical graph $(G,c_1)$.}{Does $(G,c_1)$ admit a homomorphism to $(H,c)$?}

\smallskip

\problemdec{\textsc{$T$-CSP}}{A relational structure $S$ over the same vocabulary as $T$.}{Does $S$ admit a homomorphism to $T$?}

\smallskip









\problemdec{\textsc{$k$-SAT}}{A pair $(X,C)$ where $X$ is a set of Boolean variables and $C$ is a set of $k$-tuples of literals of $X$, that is, variables of $X$ or their negation.}{Is there a truth assignment $A:X\to\{0,1\}$ such that each clause of $C$ contains at least one true literal?}

\smallskip

\problemdec{\textsc{NAE $k$-SAT}}{A pair $(X,C)$ where $X$ is a set variables and $C$ is a set of $k$-tuples of variables of $X$.}{Is there a partition of $X$ into two classes such that each clause of $C$ contains at least one variable in each class?}

It is a folklore result that \textsc{$2$-SAT} is polynomial-time solvable, a fact for example observed in~\cite{K67}. On the other hand, \textsc{$3$-SAT} is NP-complete~\cite{K72}, and \textsc{NAE $3$-SAT} is NP-complete as well~\cite{M98} (even if the input formula contains no negated variables).

\subsection{Bipartite graphs}

We now give several facts that are useful when working with homomorphisms of bipartite graphs.

\begin{observation}\label{obs:bipartite-hom}
Let $H$ be a bipartite graph with parts $A,B$. If $\phi:G\to H$ is a homomorphism of $G$ to $H$, then $G$ must be bipartite. Moreover, if $G$ and $H$ are connected, then $\phi^{-1}(A)$ and $\phi^{-1}(B)$ are the two parts of $G$.
\end{observation}

The next proposition shows that for bipartite target graphs, we may assume (at the cost of doubling the number of colours) that no two vertices from two different parts of the bipartition are coloured with the same colour.

\begin{proposition}\label{prop:bipartite-hom}
Let $(H,c)$ be a connected tropical bipartite graph with parts $A,B$, and assume that vertices in $A$ and $B$ are coloured by $c$ with colours in set $C_A$ and $C_B$, respectively. Let $c'$ be the colouring with colour set $(C_A\times 0)\cup (C_B\times 1)$ obtained from $c$ with $c'(x)=(c(x),0)$ if $x\in A$ and $c'(x)=(c(x),1)$ if $x\in B$. If \PBCOL{$(H,c')$} is polynomial-time solvable, then \PBCOL{$(H,c)$} is polynomial-time solvable.
\end{proposition}
\begin{proof}
Let $(G,c_1)$ be a bipartite tropical graph. We may assume $G$ is connected since the complexity of \PBCOL{$(H,c)$} and \PBCOL{$(H,c')$} stays the same for connected inputs. Let $c_1'$ and $c_1''$ be the colourings obtained from $c_1$ by performing a similar modification as for $c'$: $c_1'(x)=(c_1(x),0)$ if $x\in A$ and $c_1'(x)=(c_1(x),1)$ if $x\in B$, and $c_1''(x)=(c_1(x),1)$ if $x\in A$ and $c_1''(x)=(c_1(x),0)$ if $x\in B$. Now it is clear, by Observation~\ref{obs:bipartite-hom}, that $(G,c_1)\to (H,c)$ if and only if either $(G,c_1')\to (H,c')$ or $(G,c_1'')\to (H,c')$. Since the latter condition can be checked in polynomial time, the proof is complete.
\end{proof}


\subsection{Generic lemmas for polynomiality}

We now prove several generic lemmas that will be useful to prove that a specific \PBCOL{$(H,c)$} problem is polynomial-time solvable.


\begin{definition}
Let $(H,c)$ be a tropical graph. A vertex of $(H,c)$ is a
\emph{forcing vertex} if all its neighbours are coloured with distinct
colours.
\end{definition}

This is a useful concept since in any mapping of a tropical
graph $(G,c')$ to a target containing a forcing vertex $x$, if a
vertex of $G$ is mapped to $x$, then the mapping of all its neighbours
is forced. We have the following immediate application:

\begin{lemma}\label{lemm:all-forcing}
Let $(H,c)$ be a tropical graph. If all vertices of $H$ are
forcing vertices, then \PBCOL{$(H,c)$} is polynomial-time solvable.
\end{lemma}
\begin{proof}
Choose any vertex $x$ of the instance $(G,c_1)$, and map it to any
vertex of $(H,c)$ with the same colour. Once this choice is made, the
mapping for the whole connected component of $x$ is forced. Hence, try
all $O(|V(H)|)$ possibilities to map $x$, and repeat this for every
connected component of $G$. The tropical graph $(G,c_1)$ is a YES-instance if and only if
every connected component admits a mapping.
\end{proof}

\begin{lemma}[\textsc{$2$-SAT}]\label{lemm:2SAT}
Let $(H,c)$ be a tropical graph and let $\{S_1,\ldots,S_k\}$ be a collection of independent sets of $H$, each of size at most~$2$. Assume that for every tropical graph $(G,c_1)$ admitting a homomorphism
to $(H,c)$, there exists a partition $\mathcal P=P_1,\ldots,P_\ell$ of
$V(G)$ into $\ell\leq k$ sets and a homomorphism $f:(G,c_1)\to (H,c)$
such that for every $i\in\{1,\ldots,\ell\}$, there is a
$j=j(i)\in\{1,\ldots,k\}$ such that all vertices of $P_i$ map to
vertices of $S_j$. Then \PBCOL{$(H,c)$} is polynomial-time solvable.
\end{lemma}
\begin{proof}
We reduce \PBCOL{$(H,c)$} to \textsc{$2$-SAT}. For every set $S_i$, if
$S_i$ contains only one vertex $s$, $s$ represents TRUE. If $S_i$
contains two vertices $s,s'$, one of them represents TRUE, the other
FALSE (note that if some vertex belongs to two distinct sets $S_i$ and
$S_j$, it is allowed to represent, say, FALSE with respect to $S_i$
and TRUE with respect to $S_j$). Now, given an instance $(G,c_1)$ of
\PBCOL{$(H,c)$}, we build a \textsc{$2$-SAT} formula over variable set
$V(G)$ that is satisfiable if and only if $(G,c_1)\to (H,c)$, as
follows.

For every edge $xy$ of $G$, assume that in $f$, $x$ is mapped to a
vertex of $S_i$ and $y$ is mapped to a vertex of $S_j$ (necessarily if
$(G,c_1)\to (H,c)$ we have $i\neq j$ since $S_i,S_j$ induce
independent sets). Let $F_{xy}$ be a disjunction of conjunctive
2-clauses over variables $x,y$. For every edge $uv$ between a vertex
$u$ in $S_i$ and a vertex $v$ in $S_j$, depending on the truth value
assigned to $u$ and $v$, add to $F_{xy}$ the conjunctive clause that
would be true if $x$ is assigned the truth value of $u$ and $y$ is
assigned the truth value of $v$. For example: if $u=FALSE$ and
$v=TRUE$ add the clause $(\overline{x}\wedge y)$. When $F_{xy}$ is
constructed, transform it into an equivalent conjunction of
disjunctive clauses and add it to the constructed \textsc{$2$-SAT}
formula. Now, by the construction, if the formula is satisfiable we
construct a homomorphism by mapping every vertex $x$ to the vertex of
the corresponding set $S_i$ that has been assigned the same truth
value as $x$ in the satisfying assignment. By construction it is clear
that this is a valid mapping. On the other hand, if the formula is not
satisfiable, there is no homomorphism of $(G,c_1)$ to $(H,c)$
satisfying the conditions, and hence there is no homomorphism at all.
\end{proof}

As a corollary of Lemma~\ref{lemm:2SAT} and Proposition~\ref{prop:bipartite-hom}, we obtain the following lemma:

\begin{lemma}\label{lemm:2SAT-application}
If $(H,c)$ is a bipartite tropical graph where each colour is used at most twice, then \PBCOL{$(H,c)$} is polynomial-time solvable.
\end{lemma}


Given a set $S$ of vertices, the \emph{boundary} $B(S)$ is the set of vertices in $S$ that have a neighbour out of $S$.


\begin{lemma}\label{lemm:distinct-vertex-boundary}
Let $(H,c)$ be a tropical graph containing a connected subgraph
$S$ of forcing vertices such that:\\
(a) every vertex in $B(S)$ is coloured with a distinct colour (let $C(S)$ be the set of colours given to vertices in $B(S)$), and
(b) no colour of $C(S)$ is present in $V(H)\setminus S$.\\
If \PBlistCOL{$(H-S)$} is polynomial-time solvable, then \PBCOL{$(H,c)$} is polynomial-time solvable.
\end{lemma}
\begin{proof}
Let $\overline{S}=V(H)\setminus S$. Let $(G,c_1)$ be an instance of
\PBCOL{$(H,c)$}. Consider an arbitrary vertex $v$ of $G$ with
$c_1(v)=i$. Then, $v$ must be mapped to a vertex
coloured~$i$. For every possible choice of mapping $v$, we will
construct one instance of \PBlistCOL{$(H-S)$}. To construct an
instance from such a choice, we first partition $V(G)$ into two sets:
the set $V_S$ containing the vertices that must map to vertices in $S$
(and their images are determined), and the set $V_{\overline{S}}$
containing the vertices that must map to vertices of $\overline{S}$. We
now distinguish two basic cases, that will be repeatedly applied during
the construction.

\noindent\textbf{Case~1: vertex \boldmath{$v$} is mapped to a vertex in \boldmath{$S$}.}  If $v$ has been mapped to a vertex $x$ of $S$,
since $x$ is a forcing vertex, the mapping of all neighbours of $v$ is
determined (anytime there is a conflict we return NO for the specific
instance under construction). We continue to propagate the forced
mapping as much as possible (i.e. as long as the forced images belong to
$S$) within a connected set of $G$ containing $v$. This yields a
connected set $C_v$ of vertices of $G$ whose mapping is determined,
and whose neighbourhood $N_v=N(C_v)\setminus C_v$ consists of vertices
each of which must be mapped to a determined vertex of
$\overline{S}$. We add $C_v$ to $V_S$. We now remove the set $C_v$
from $G$ and repeat the procedure for all vertices of $N_v$ using
Case~2.

\noindent\textbf{Case~2: vertex \boldmath{$v$} is mapped to a vertex
  in \boldmath{$\overline{S}$}.}  We perform a BFS search on the
remaining vertices in $G$, until we have computed a maximal connected
set $C_v$ of vertices containing $v$ in which no vertex is coloured
with a colour in $C(S)$. Then, for every vertex $x$ of $C_v$ with a
neighbour $y$ that is coloured~$i$ ($i\in C(S)$), by Property~(a) we
know that $y$ must be mapped to a vertex in $B(S)$, and moreover the
image of $y$ is determined by colour~$i$. Hence the neighbourhood
$N_v=N(C_v)\setminus C_v$ has only vertices whose mapping is
determined. We add $C_v$ to set $V_{\overline{S}}$ and apply Case~1 to
every vertex in $N_v$.

\noindent\textbf{End of the procedure.} Once $V(G)$ has been
partitioned into $V_S$ and $V_{\overline{S}}$ (where the mapping of all
vertices in $V_S\cup N(V_S)$ is fixed),  we
can reduce this instance to a corresponding instance of
\PBlistCOL{$(H-S)$}. 

In total, $(G,c_1)$ is a YES-instance if and only if at least one of the
$O(|V(G)|)$ constructed instances of \PBlistCOL{$(H-S)$} is a YES-instance.
\end{proof}

The next lemma is similar to Lemma~\ref{lemm:distinct-vertex-boundary}
but now the boundary is distinguished using edges.

\begin{lemma}\label{lemm:distinct-edge-boundary}
Let $(H,c)$ be a tropical graph containing a connected subgraph
$S$ of forcing vertices with boundary $B=B(S)$ and $N=N(B)\setminus S$. Assume that the following properties hold:\\
(a) for every pairs $xy$, $x'y'$ of distinct edges of $B\times N$, we have $(c(x),c(y))\neq (c(x'),c(y'))$, and\\
(b) for every edge $xy$ of $B\times N$, there is no edge in $(H-S)\times (H-S)$ whose endpoints are coloured $c(x)$ and $c(y)$.
If \PBlistCOL{$(H-S)$} is polynomial-time solvable, then \PBCOL{$(H,c)$} is
polynomial-time solvable.
\end{lemma}
\begin{proof}
The proof is almost the same as the one of
Lemma~\ref{lemm:distinct-vertex-boundary}, except that now, while
computing an instance of \PBlistCOL{$(H-S)$}, the distinction between
$V_S$ and $V_{\overline{S}}$ is determined by the edges of $B\times
N$.
\end{proof}

The next lemma identify some unique features of a tropical graph to simplify the problem into a list-homomorphism problem.

\begin{definition}
A \emph{Unique Tropical Feature} in a tropical graph $(H,c)$ is a vertex or an edge of $H$ that satisfies one of the following conditions.
\begin{itemize}
\item[Type 1.] A vertex $u$ of $H$ whose colour class is $\{u\}$.
\item[Type 2.] An edge $uv$ of $H$ such that there is no other edge in $H$ whose vertices are coloured $c(u)$ and $c(v)$, respectively.
\item[Type 3.] A vertex $u$ of $H$ such that $N(u)$ is monochromatic in $(H,c)$ with colour $s$, and every vertex coloured $s$ that does not belong to $N(u)$ has no neighbour coloured with $c(u)$.
\item[Type 4.] A forcing vertex $u$ of $H$ such that for each pair $v,w$ of distinct vertices in $N(u)$, there is no path $v'u'w'$ in $H-u$ with $c(v) = c(v')$, $c(u)=c(u')$ and $c(w)=c(w')$. 
\end{itemize}
\end{definition}

\begin{definition}
Let $(H,c)$ be a tropical graph and $S$ a set of Unique Tropical Features of $(H,c)$. $S$ is partitioned into four sets as $S=S_1 \cup S_2 \cup S_3 \cup S_4$, where $S_i$ is the set of unique tropical features of type $i$ in $S$. We define $H(S)$ as follows : $V(H(S)) = (V(H)\cup \{u_v| u\in S_4, v\in N(u) \})\setminus (S_1 \cup S_3 \cup S_4)$ and $E(H(S)) = (E(H[V(H(S))])\setminus S_2)\cup \{u_vv| u\in S_4, v\in N(u)\}$.
\end{definition}

In other words, $H(S)$ is the graph obtained from $H$ by removing unique tropical features of type $1$, $2$, and $3$, and for each unique tropical feature $u$ of type $4$, replacing $N[u]$ by $d(u)$ pending edges.

\begin{lemma}\label{lemm:unique-feature}
Let $(H,c)$ be a tropical graph and $S$ a set of unique tropical features of $(H,c)$. If \PBlistCOL{$(H(S))$} is polynomial-time solvable, then \PBCOL{$(H,c)$} is polynomial-time solvable.
\end{lemma}

\begin{proof}
Let $(G,c')$ be an instance of \PBCOL{$(H,c)$}. We are going to construct a graph $G'$ and associate to each vertex of $G'$ a list of vertices of $H(S)$ such that there is a list-homomorphism from $G'$ to $H(S)$ (with respect to these lists) if and only if there is a tropical homomorphism of $(G,c')$ to $(H,c)$. We proceed with sequential modifications, by considering the unique tropical features of $S$ one by one.

First, we can see the instance $(G,c')$ of \PBCOL{$(H,c)$} as an instance of \PBlistCOL{$H$} by giving to each vertex $u$ in $G$ the list $L(u)$ of vertex in $H$ coloured $c'(u)$. If at any point in the following, we update the list of a vertex to be empty, we can conclude that there is no tropical homomorphism between $(G,c')$ and $(H,c)$. 

For each unique tropical feature $u$ of type $1$ in $S$, there is a colour $s$ such that only the vertex $u$ is coloured $s$ in $(H,c)$. Every vertex in $(G,c')$ coloured $s$ must be mapped to $u$ and has a list of size at most one. For each vertex $v$ in $(G,c')$ coloured $s$, we update the list of each of its neighbours $w$ such that $L(w)$ becomes $L(w)\cap N(u)$. We can then delete $v$ from $(G,c')$ and forget $L(v)$ without affecting the existence of a list-homomorphism. Indeed, if a homomorphism exists, then it must map each neighbour of $v$ to a neighbour of $u$. Moreover, there is no other vertex of $(G,c')$ that can be mapped to $u$.

For each unique tropical feature $uv$ of type $2$ in $S$, there is no other edge than $uv$ in $H$ such that the colour of its vertices are $c(u)$ and $c(v)$. Every edge in $(G,c')$ whose vertices are coloured $c(u)$ and $c(v)$ must be mapped to $uv$. For each edge $xy$ in $(G,c')$ such that $c'(x)= c(u)$ and $c'(y)=c(v)$, we update the list of $x$ and $y$ such that $L(x)$ becomes $L(x)\cap \{u\}$ and $L(y)$ becomes $L(y)\cap \{v\}$. We can then delete the edge $uv$ from $(G,c')$ without changing the existence of a list-homomorphism. Indeed, if a homomorphism exists, it must map $x$ to $u$ and $y$ to $v$. Again, there is no other edge of $(G,c')$ that can be mapped to $uv$.

For each unique tropical feature $u$ of type $3$ in $S$, $N(u)$ is monochromatic in $(H,c)$ of colour $s$ and any vertex coloured $s$ with a neighour coloured $c(u)$ must belong to $N(u)$. Let $v$ be a vertex of $G$ such that $c(v)=c(u)$ and $N(v)$ is monochromatic in $(G,c')$ of colour $s$. Then, we can assume that $v$ is mapped to $u$. Indeed, in every tropical homomorphism of $(G,c')$ to $(H,c)$, if $v$ is not mapped to $u$, it is mapped to a vertex at distance~$2$ from $u$, and one obtains another valid tropical homomorphism by only changing the mapping of $v$ to $u$. For each such vertex $v$, we update the list of its neighbours $w$ such that $L(w)$ becomes $L(w)\cap N(u)$. We can then delete $v$ from $(G,c')$ without affecting the existence of a list-homomorphism. Indeed, if a homomorphism exists, it maps every neighbour of $v$ to a neighbour of $u$. Moreover, there no other vertex of $(G,c')$ can be mapped to $u$.

Finally, let $u$ be a vertex of type $4$ in $S$. Thus, by the definition of type $4$, for each $v, w \in N(u)$, there is no other path $v'u'w'$ in $H$ such that $c(v)=c(v')$, $c(u)=c(u')$ and $c(w)=c(w')$. Furthermore, since $u$ is a forcing vertex, we have $c(v)\neq c(w)$ for any two neighbours $v$ and $w$ of $u$.

Let $x$ be a vertex of $G$ such that $c'(x)=c(u)$ and such that at least two neighbours of $x$ are of colours $c(v)$ or $c(w)$, one of each. Then, as $x$ is of type $4$, any homomorphism of $(G, c')$ to $(H,c)$ must map all such vertices $x$ to $u$. Remove all such vertices from $G$ and let $(G', c')$ be the remaining tropical graph. For any vertex $y$ of $G'$ if it is of colour $c(u)$, it may then either map to another vertex of this colour, or all its neighbours must map a same neighbour of $u$. Let $(H_1, c)$ be a tropical graph obtained from $(H, c)$ by removing the vertex $u$, and then adding one new vertex for each vertex in $N_H(u)$ and assigning the colour $c(u)$ to it. It follows that $(G',c')$ admits a homomorphism to $(H',c)$ if and only 
$(G,c')$ admits a homomorphism to $(H,c)$, proving our claim.

In conclusion, we have built an instance $(G', L)$ of \PBlistCOL{$H(S)$} that maps to $H(S)$ if and only if $(G,c')$ maps to $(H,c)$, thus proving our claim. 
We remark, furthermore, that these changes used to introduced $(G', L)$ and $H(S)$ are compatible even between different types of vertices, thus we may allow $S$ to contain a combination of such vertices. However, in this work we will only consider sets $S$ whose elements are all of a same type.

\end{proof}

\section{\PBCOL{$(H,c)$} and the Dichotomy Conjecture}\label{sec:dicho-CSP}

Since each \PBCOL{$(H,c)$} problem is a CSP, the Feder--Vardi Dichotomy Conjecture (Conjecture~\ref{conj:CSP-dicho}) would imply a complexity dichotomy for the class of \PBCOL{$(H,c)$} problems. As we mentioned before a proof of the conjecture has been recently announced, thus every \PBCOL{$(H,c)$} is either polynomial time solvable or it is an NP-complete problem. Here we point out that an independent proof even on a very restricted set of $(H,c)$ would also prove the original conjecture.


Following the construction of Feder and Vardi (\cite[Theorem~10]{FV98}) and based on its exposition in the book by Hell and Ne\v{s}et\v{r}il~\cite[Theorem~5.14]{HNbook}, one can modify their gadgets to prove a similar statement for the class of $2$-tropical bipartite graph homomorphism problems. 
\shortpaper{The proof being very similar to the proofs in~\cite{BFHN15,FV98,HNbook}, we skip it here and refer to our manuscript~\cite{fullversion} instead.
}
%
%
%
%

\begin{theorem}\label{thm:CSP}
For each CSP template $T$ there is a $2$-coloured graph $(H,c)$ such that \PBCOL{$(H,c)$} 
and \textsc{$T$-CSP} are polynomially equivalent. Moreover, $(H,c)$ can be chosen to be bipartite 
and homomorphic to a $2$-coloured forcing path.
\end{theorem}
\longpaper{
\begin{proof}
We follow the proof of Theorem~5.14 in the book~\cite{HNbook} proving a similar statement for digraph 
homomorphism problems. The structure of the proof in~\cite{HNbook} is as follows. First, one shows that 
for each CSP template $T$, there is a bipartite graph $H$ such that the \textsc{$T$-CSP} problem and 
the \textsc{$H$-Retraction} problem are polynomially equivalent. Next, it is shown that for each 
bipartite graph $H$ there is a digraph $H'$ such that \textsc{$H$-Retraction} and 
\textsc{$H'$-Retraction} are polynomially equivalent. Finally it is observed that $H'$ is a core and 
thus \textsc{$H'$-Retraction} and \PBCOL{$(H',c)$} are polynomially equivalent. 
We adapt this proof to the case of $2$-tropical graph homomorphism problems.

The construction of $H'$ from $H$ in~\cite{HNbook} is through the use of so-called zig-zag paths. 
In our case, we replace these zig-zag paths by specific $2$-coloured graphs that play the same role. This will allow us to construct 
a $2$-coloured graph $H'$ from a bipartite graph $H$ such that \textsc{$H$-Retraction} and 
\textsc{$H'$-Retraction} are polynomially equivalent. Our paths will have black vertices denoted by 
$B$ and white vertices denoted by $W$. Hence the path $WB^4W^4B$ consists of one white vertex, four black vertices, 
four white vertices and a black vertex. The maximal monochromatic subpaths are called \emph{runs}. 
Thus the above path is the concatenation of four runs: the first and last of length~$1$, the middle two of length~$4$.

Given an odd integer $\ell$, we construct a path $P$ consisting of $\ell$ runs. 
The first and the last run each consist of a single white vertex. The interior runs are 
of length four. We denote that last (rightmost) vertex of $P$ by $0$. From $P$ we 
construct $\ell-2$ paths $P_1, \dots, P_{\ell-2}$. Path $P_i$ ($i=1, 2, \dots, \ell-2$) 
is obtained from $P$ by replacing the $i^{th}$ run of length four with a run of length~$2$. 
We denote the rightmost vertex of $P_i$ by $i$. 

Similarly, for an even integer $k$, we construct a second family of paths $Q$ and 
$Q_j$, ($j=1, 2, \dots, k-2$). The leftmost vertex of $Q$ is $1$ and 
the leftmost vertex of $Q_j$ is $j$. The paths are described below:

$$
\begin{array}{rclcrcl}
P   & := & W\underbrace{B^4 W^4 \cdots W^4 B^4}_{\ell-2}W & \hspace{0.5cm} &
Q   & := & W\underbrace{B^4 W^4 \cdots B^4 W^4}_{k-2}B \\
P_i & := & W\underbrace{B^4 \cdots W^4}_{i-1}B^2\underbrace{W^4 \cdots B^4}_{\ell-i-2}W \hspace{1em}\mbox{($i$ odd)} & &
Q_j & := & W\underbrace{B^4 \cdots W^4}_{j-1}B^2\underbrace{W^4 \cdots W^4}_{k-j-2}B \hspace{1em}\mbox{($j$ odd)} \\
P_i & := & W\underbrace{B^4 \cdots B^4}_{i-1}W^2\underbrace{B^4 \cdots B^4}_{\ell-i-2}W \hspace{1em}\mbox{($i$ even)} & &
Q_j & := & W\underbrace{B^4 \cdots B^4}_{j-1}W^2\underbrace{B^4 \cdots W^4}_{k-j-2}B \hspace{1em}\mbox{($j$ even)} \\
\end{array}
$$

We observe the following (c.f. page 156 of~\cite{HNbook}):
\begin{enumerate}
 \item The paths $P$ and $P_i$ ($i=1,2,\dots \ell-2$) each admit a homomorphism \emph{onto} a \emph{$2$-colour forcing path} of length~$2\ell - 1$, (that is, a path consisting of one run of length~$1$, $\ell - 2$ runs each of length~$2$ and a final run of length~$1$: $WBBWWB\cdots W$).
 \item The paths $Q$ and $Q_j$ ($j=1,2,\dots k-2$) each admit a homomorphism onto a $2$-colour forcing path of length~$2k-1$. 
 \item $P_i \to P_{i'}$ implies $i=i'$.
 \item $Q_j \to Q_{j'}$ implies $j=j'$.
 \item $P \to P_i$ for all $i$.
 \item $Q \to Q_j$ for all $j$.
 \item if $X$ is a $2$-tropical graph and $x$ is a vertex of $X$ such that $f: X \to P_i$ and $f': X \to P_{i'}$ for $i\neq i'$ with $f(x) = i$ and $f'(x) = i'$, then there is a homomorphism $F:X \to P$ with $F(x)=0$. 
 \item if $Y$ is a $2$-tropical graph and $y$ is a vertex of $Y$ such that $f: Y \to Q_j$ and $f': Y \to Q_{j'}$ for $j \neq j'$ with $f(y) = j$ and $f'(y) = j'$, then there is a homomorphism $F:Y \to Q$ with $F(y)=1$.
\end{enumerate}

We note that $2$-colour forcing paths in $2$-tropical graphs can be used to define \emph{height} analogously to height in directed acyclic graphs. More precisely, suppose $G$ is a connected $2$-tropical graph that admits a homomorphism onto a $2$-colour forcing path, say $FP$, of even length. Let the vertices of $FP$ be $h_0, h_1, \dots, h_{2t}$. Observe that each vertex in the path has at most one white neighbour and at most one black neighbour. Thus once a single vertex $u$ in $G$ is mapped to $FP$, the image of each neighbour of $u$ is uniquely determined and by connectivity, the image of all vertices is uniquely determined. In particular, as $G$ maps \emph{onto} $FP$, there is exactly one homomorphism of $G$ to the path. (More precisely, if the path has length congruent to 0 modulo 4, there is an automorphism that reverses the path. In this case there are two homomorphisms that are equivalent up to the reversing.) We then observe that if $g:G \stackrel{onto}{\to} FP$, $h:H \to FP$, and $f: G \to H$, then for all vertices $u \in V(G)$, $g(u) = h(f(u))$. This allows us to define the height of $u \in V(G)$ to be $h_i$ when $g(u)=h_i$. Specifically, vertices at height $h_i$ in $G$ must map to vertices at height $h_i$ in $H$. 

For each problem $T$ in CSP there is a bipartite graph $H$ such that \textsc{$T$-CSP} and \textsc{$H$-Retraction} are equivalent~\cite{FV98,HNbook}. Let $H$ be a bipartite graph with parts $(A,B)$, with $A=\{a_1,\dots,a_{|A|}\}$ and $B=\{b_1,\dots,b_{|B|}\}$. Let $\ell$ (respectively $k$) be the smallest odd (respectively even) integer greater than or equal to $|A|$ (respectively $|B|$). To each vertex $a_i \in A$ attach a copy of $P_i$ identifying $i$ in $P_i$ with $a_i$ in $A$. Colour all original vertex of $H$ white. To each vertex $b_j \in B$ attach a copy of $Q_j$ identifying $j$ in $Q_j$ with $b_j$ in $B$.  Call the resulting $2$-tropical graph $(H',c)$. See Figure~\ref{fig:rbtarget} for an illustration.

Let $G$ be an instance of \textsc{$H$-Retraction}. In particular, we may assume without loss of generality that $H$ is a subgraph of $G$, $G$ is connected, and $G$ is bipartite. We colour the original vertices of $G$ white. Let $(A',B')$ be the partite classes of $G$ where $A \subseteq A'$ and $B \subseteq B'$. To each vertex $v$ of $A' \backslash A$, we attach a copy of $P$, identifying $v$ and $0$. To the vertices of $A \cup B$, we attach paths $P_i$ and $Q_j$ as described above to create a copy of $H'$. Call the resulting $2$-tropical graph $(G',c')$. In particular, note that $(G',c')$ and $(H',c')$ both map onto a $2$-colour forcing path of length $2\ell+2k-1$. The (original) vertices of $G$ and $H$ are at height $2\ell-1$ and $2\ell$ for colour classes $A$ and $B$ respectively. In particular, by the eight above properties, under any homomorphism $f:G' \to H'$ the restriction of $f$ to $G$ must map onto $H$ with vertices in $A'$ mapping to $A$ and vertices in $B'$ mapping to $B$.

Using the eight properties of the paths above and following the proof of Theorem~5.14 in~\cite{HNbook}, we conclude that $G$ is a YES instance of \textsc{$H$-Retraction} if and only if $(G',c')$ is a YES instance of \textsc{$(H',c)$-Retraction}.

On the other hand, let $(G',c')$ be an instance of \textsc{$(H',c)$-Retraction}. We sketch the proof from~\cite{HNbook}. We observe that $(G',c')$ must map to a $2$-colour forcing path of length $2\ell+2k-1$. The two levels of $G'$ corresponding to $H$ induce a bipartite graph (with white vertices) which we call $G$. The components of $G'-E(G)$ fall into two types: those which map to lower levels and those that map to higher levels than $G$. Let $C_t$ be a component that maps to a lower level. After required identifications we may assume $C_t$ contains only one vertex from $G$ (say $v$) and $C_t$ must map to some $P_i$. If $P_i$ is the unique $P_i$ path to which $C_t$ maps, then we modify $G'$ by identifying $v$ and $i$. Otherwise, $C_t$ maps to two paths and (by the properties~$5$--$8$) hence to all paths. The resulting graph $(G',c')$ retracts to $(H',c)$ if and only if $G$ retracts to $H$.
\end{proof}

\tikzset{
  bigblue/.style={circle, draw=blue!80,fill=blue!40,thick, inner sep=1.5pt, minimum size=5mm},
  bigred/.style={circle, draw=red!80,fill=red!40,thick, inner sep=1.5pt, minimum size=5mm},
  bigblack/.style={circle, draw=black!100,fill=black!40,thick, inner sep=1.5pt, minimum size=5mm},
  bluevertex/.style={circle, draw=blue!100,fill=blue!100,thick, inner sep=0pt, minimum size=2mm},
  redvertex/.style={circle, draw=red!100,fill=red!100,thick, inner sep=0pt, minimum size=2mm},
  blackvertex/.style={circle, draw=black!100,fill=black!100,thick, inner sep=0pt, minimum size=2mm},  
  whitevertex/.style={circle, draw=black!100,fill=white!100,thick, inner sep=0pt, minimum size=2mm},  
  smallblack/.style={circle, draw=black!100,fill=black!100,thick, inner sep=0pt, minimum size=1mm},
	smallwhite/.style={circle, draw=black!100,fill=white!100,thick, inner sep=0pt, minimum size=1mm},
}

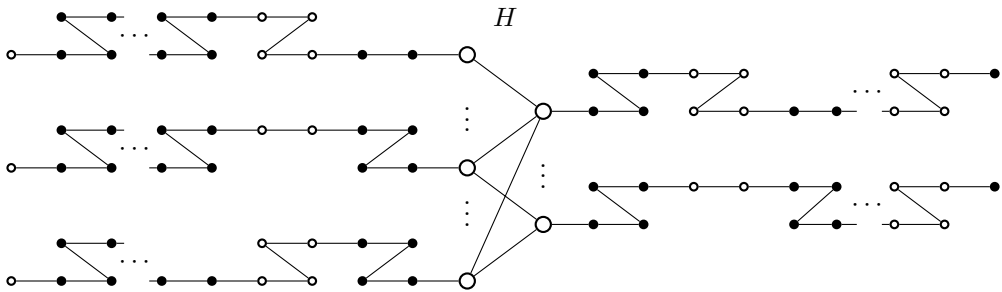
\begin{figure}[ht!]
\begin{center}
\begin{tikzpicture}
\node[smallwhite] (p1) at (0,4) {};
\node[smallblack] (p2) at (0.66,4) {};
\node[smallblack] (p3) at (1.32,4) {};
\node[smallblack] (p4) at (0.66,4.5) {};
\node[smallblack] (p5) at (1.32,4.5) {};
\draw[black] (p1)--(p2)--(p3)--(p4)--(p5)--(1.485,4.5);
\node at (1.65,4.25) {$\cdots$};
\node[smallblack] (p6) at (1.98,4) {};
\node[smallblack] (p7) at (2.64,4) {};
\node[smallblack] (p8) at (1.98,4.5) {};
\node[smallblack] (p9) at (2.64,4.5) {};
\node[smallwhite] (p10) at (3.3,4.5) {};
\node[smallwhite] (p11) at (3.96,4.5) {};
\node[smallwhite] (p12) at (3.3,4) {};
\node[smallwhite] (p13) at (3.96,4) {};
\node[smallblack] (p14) at (4.62,4) {};
\node[smallblack] (p15) at (5.28,4) {};
\node[whitevertex] (pa1) at (6,4) {};
\draw[black] (1.815,4)--(p6)--(p7)--(p8)--(p9)--(p10)--(p11)--(p12)--(p13)--(p14)--(p15)--(pa1);


\begin{scope}[yshift=-1.5cm]
\node[smallwhite] (p1) at (0,4) {};
\node[smallblack] (p2) at (0.66,4) {};
\node[smallblack] (p3) at (1.32,4) {};
\node[smallblack] (p4) at (0.66,4.5) {};
\node[smallblack] (p5) at (1.32,4.5) {};
\draw[black] (p1)--(p2)--(p3)--(p4)--(p5)--(1.485,4.5);
\node at (1.65,4.25) {$\cdots$};
\node[smallblack] (p6) at (1.98,4) {};
\node[smallblack] (p7) at (2.64,4) {};
\node[smallblack] (p8) at (1.98,4.5) {};
\node[smallblack] (p9) at (2.64,4.5) {};
\node[smallwhite] (p10) at (3.3,4.5) {};
\node[smallwhite] (p11) at (3.96,4.5) {};
\node[smallblack] (p12) at (4.62,4.5) {};
\node[smallblack] (p13) at (5.28,4.5) {};
\node[smallblack] (p14) at (4.62,4) {};
\node[smallblack] (p15) at (5.28,4) {};
\node[whitevertex] (pa2) at (6,4) {};
\draw[black] (1.815,4)--(p6)--(p7)--(p8)--(p9)--(p10)--(p11)--(p12)--(p13)--(p14)--(p15)--(pa2);
\end{scope}

\begin{scope}[yshift=-3cm]
\node[smallwhite] (p1) at (0,4) {};
\node[smallblack] (p2) at (0.66,4) {};
\node[smallblack] (p3) at (1.32,4) {};
\node[smallblack] (p4) at (0.66,4.5) {};
\node[smallblack] (p5) at (1.32,4.5) {};
\draw[black] (p1)--(p2)--(p3)--(p4)--(p5)--(1.485,4.5);
\node at (1.65,4.25) {$\cdots$};
\node[smallblack] (p6) at (1.98,4) {};
\node[smallblack] (p7) at (2.64,4) {};
\node[smallwhite] (p8) at (3.3,4) {};
\node[smallwhite] (p9) at (3.96,4) {};
\node[smallwhite] (p10) at (3.3,4.5) {};
\node[smallwhite] (p11) at (3.96,4.5) {};
\node[smallblack] (p12) at (4.62,4.5) {};
\node[smallblack] (p13) at (5.28,4.5) {};
\node[smallblack] (p14) at (4.62,4) {};
\node[smallblack] (p15) at (5.28,4) {};
\node[whitevertex] (pa3) at (6,4) {};
\draw[black] (1.815,4)--(p6)--(p7)--(p8)--(p9)--(p10)--(p11)--(p12)--(p13)--(p14)--(p15)--(pa3);
\end{scope}

\begin{scope}[yshift=-0.75cm]
\node[whitevertex] (qb1) at (7,4) {};
\node[smallblack] (q11) at (7.66,4) {};
\node[smallblack] (q12) at (8.32,4) {};
\node[smallblack] (q13) at (7.66,4.5) {};
\node[smallblack] (q14) at (8.32,4.5) {};
\node[smallwhite] (q15) at (8.98,4.5) {};
\node[smallwhite] (q16) at (9.64,4.5) {};
\node[smallwhite] (q17) at (8.98,4) {};
\node[smallwhite] (q18) at (9.64,4) {};
\node[smallblack] (q19) at (10.3,4) {};
\node[smallblack] (q20) at (10.86,4) {};
\draw[black] (qb1)--(q11)--(q12)--(q13)--(q14)--(q15)--(q16)--(q17)--(q18)--(q19)--(q20)--(11.125,4);
\node at (11.29,4.25) {$\cdots$};
\node[smallwhite] (q21) at (11.62,4) {};
\node[smallwhite] (q22) at (12.28,4) {};
\node[smallwhite] (q23) at (11.62,4.5) {};
\node[smallwhite] (q24) at (12.28,4.5) {};
\node[smallblack] (q25) at (12.94,4.5) {};
\draw[black] (11.455,4)--(q21)--(q22)--(q23)--(q24)--(q25);

\end{scope}

\begin{scope}[yshift=-2.25cm]
\node[whitevertex] (qb2) at (7,4) {};
\node[smallblack] (q11) at (7.66,4) {};
\node[smallblack] (q12) at (8.32,4) {};
\node[smallblack] (q13) at (7.66,4.5) {};
\node[smallblack] (q14) at (8.32,4.5) {};
\node[smallwhite] (q15) at (8.98,4.5) {};
\node[smallwhite] (q16) at (9.64,4.5) {};
\node[smallblack] (q17) at (10.3,4.5) {};
\node[smallblack] (q18) at (10.86,4.5) {};
\node[smallblack] (q19) at (10.3,4) {};
\node[smallblack] (q20) at (10.86,4) {};
\draw[black] (qb2)--(q11)--(q12)--(q13)--(q14)--(q15)--(q16)--(q17)--(q18)--(q19)--(q20)--(11.125,4);
\node at (11.29,4.25) {$\cdots$};
\node[smallwhite] (q21) at (11.62,4) {};
\node[smallwhite] (q22) at (12.28,4) {};
\node[smallwhite] (q23) at (11.62,4.5) {};
\node[smallwhite] (q24) at (12.28,4.5) {};
\node[smallblack] (q25) at (12.94,4.5) {};
\draw[black] (11.455,4)--(q21)--(q22)--(q23)--(q24)--(q25);
\end{scope}

\draw[black] (pa1)--(qb1)--(pa2)--(qb2)--(pa3) (pa3)--(qb1);

\node at (6.5,4.5) {$H$};
\node at (6,3.25) {$\vdots$};
\node at (6,2) {$\vdots$};
\node at (7,2.5) {$\vdots$};

\end{tikzpicture}
\end{center}
\caption{Construction of a $2$-tropical target $H'$ from a \textsc{$H$-Retraction} problem.}
\label{fig:rbtarget}
\end{figure}

}

\section{Minimal graphs $H$ for NP-complete \PBlistCOL{$H$}}\label{sec:list_hom}

Recall the dichotomy theorem for list homomorphism problems of Feder, Hell and Huang (Theorem~\ref{thm:listhom-table}): \PBlistCOL{$H$} is polynomial-time solvable if $H$ is bipartite and its complement is a circular arc graph, otherwise NP-complete. Alternatively, the latter class of graphs was characterized by Trotter and Moore~\cite{TM76} in terms of seven families of forbidden induced subgraphs: six infinite ones and a finite one.
\shortpaper{One of these families is the family of even cycles of length at least~$6$. The only tree in the list of forbidden graphs, which is called $G_1$ in~\cite{FHH99}, is a claw where each edge is subdivided twice.}
\longpaper{See their descriptions in Table~\ref{table}, as reproduced from~\cite{FHH99}. To concisely describe these seven families, they employ the following notation: Let $\mathcal F = \{S_i : 1 \leq i \leq k\}$ be a family of subsets of $\{1,2,\ldots,\ell\}$. Define $H_F$ to be the bipartite graph $(X,Y)$ with $X = \{x_1,x_2,\ldots,x_\ell\}$ and $Y =\{y_1,y_2,\ldots,y_k\}$ such that $x_iy_j$ is an edge if and only if $i\in S_j$. The families $\mathcal C$, $\mathcal T$, $\mathcal W$, $\mathcal D$, $\mathcal M$, $\mathcal N$ and $\mathcal G$ in Table~\ref{table} are defined in this way. Note that the graph $C_i$ in $\mathcal C$ is the cycle of length~$i$. See Figure~\ref{fig:table} for an illustration of the other families from Table~\ref{table}. Also note that $G_1$, which is a claw where each edge is subdivided twice, is the only tree in the table.

Given the above characterization, we can reformulate Theorem~\ref{thm:listhom-table} as follows.

\begin{theorem}[Restatement of Theorem~\ref{thm:listhom-table}, Feder, Hell and Huang~\cite{FHH99}]
If $H$ contains one of the graphs defined in Table~\ref{table} as an induced subgraph, then \PBlistCOL{$H$} is NP-complete. Otherwise, \PBlistCOL{$H$} is polynomial-time solvable.
\end{theorem}

\begin{table}[ht!]
\begin{tabular}{l}
\hline
$C_6 = \{\{1, 2\}, \{2, 3\}, \{3, 1\}\}$ \\ 

$C_8 = \{\{1, 2\}, \{2, 3\}, \{3, 4\}, \{4, 1\}\}$ \\ 

$C_{10} = \{\{1, 2\}, \{2, 3\}, \{3, 4\}, \{4, 5\}, \{5, 1\}\}$\\ 

\ldots\\ 

$T_1 = \{\{1, 2\}, \{2, 3\}, \{3, 4\}, \{2, 3, 5\}, \{5\}\}$\\ 

$T_2 = \{\{1, 2\}, \{2, 3\}, \{3, 4\}, \{4, 5\}, \{2, 3, 4, 6\}\{6\}\}$\\ 

$T_3 = \{\{1, 2\}, \{2, 3\}, \{3, 4\}, \{4, 5\}, \{5, 6\}, \{2, 3, 4, 5, 7\}, \{7\}\}$\\ 

\ldots\\ 

$W_1 = \{\{1, 2\}, \{2, 3\}, \{1, 2, 4\}, \{2, 3, 4\}, \{4\}\}$\\ 

$W_2 = \{\{1, 2\}, \{2, 3\}, \{3, 4\}, \{1, 2, 3, 5\}, \{2, 3, 4, 5\}, \{5\}\}$\\ 

$W_3 = \{\{1, 2\}, \{2, 3\}, \{3, 4\}, \{4, 5\}, \{1, 2, 3, 4, 6\}, \{2, 3, 4, 5, 6\}, \{6\}\}$\\ 

\ldots\\ 

$D_1 = \{\{1, 2, 5\}, \{2, 3, 5\}, \{3\}, \{4, 5\}, \{2, 3, 4, 5\}\}$\\ 

$D_2 = \{\{1, 2, 6\}, \{2, 3, 6\}, \{3, 4, 6\}, \{4\}, \{5, 6\}, \{2, 3, 4, 5, 6\}\}$\\ 

$D_3 = \{\{1, 2, 7\}, \{2, 3, 7\}, \{3, 4, 7\}, \{4, 5, 7\}, \{5\}, \{6, 7\}, \{2, 3, 4, 5, 6, 7\}\}$\\ 

\ldots\\ 

$M_1 = \{\{1, 2, 3, 4, 5\}, \{1, 2, 3\}, \{1\}, \{1, 2, 4, 6\}, \{2, 4\}, \{2, 5\}\}$\\ 
$M_2 = \{\{1, 2, 3, 4, 5, 6, 7\}, \{1, 2, 3, 4, 5\}, \{1, 2, 3\}, \{1\}, \{1, 2, 3, 4, 6, 8\}, \{1, 2, 4, 6\}, \{2, 4\}, \{2, 7\}\}$\\ 
$M_3 = \{\{1, 2, 3, 4, 5, 6, 7, 8, 9\}, \{1, 2, 3, 4, 5, 6, 7\}, \{1, 2, 3, 4, 5\}, \{1, 2, 3\}, \{1\}, \{1, 2, 3, 4, 5, 6, 8, 10\},
$\\
\hspace{240pt}$\{1, 2, 3, 4, 6, 8\}, \{1, 2, 4, 6\}, \{2, 4\}, \{2, 9\}\}$\\ 
\ldots\\ 
$N_1 = \{\{1, 2, 3\}, \{1\}, \{1, 2, 4, 6\}, \{2, 4\}, \{2, 5\}, \{6\}\}$\\ 
$N_2 = \{\{1, 2, 3, 4, 5\}, \{1, 2, 3\}, \{1\}, \{1, 2, 3, 4, 6, 8\}, \{1, 2, 4, 6\}, \{2, 4\}, \{2, 7\}, \{8\}\}$\\ 
$N_3 = \{\{1, 2, 3, 4, 5, 6, 7\}, \{1, 2, 3, 4, 5\}, \{1, 2, 3\}, \{1\}, \{1, 2, 3, 4, 5, 6, 8, 10\}, \{1, 2, 3, 4, 6, 8\},$\\
\hspace{250pt}$\{1, 2, 4, 6\}, \{2, 4\}, \{2, 9\}, \{10\}\}$\\ 
\ldots\\ 
$G_1 = \{\{1, 3, 5\}, \{1, 2\}, \{3, 4\}, \{5, 6\}\}$\\ 
$G_2 = \{\{1\}, \{1, 2, 3, 4\}, \{2, 4, 5\}, \{2, 3, 6\}\}$\\ 
$G_3 = \{\{1, 2\}, \{3, 4\}, \{5\}, \{1, 2, 3\}, \{1, 3, 5\}\}$\\ 
\hline
\end{tabular}
\caption{Six infinite families $\mathcal C$, $\mathcal T$, $\mathcal W$, $\mathcal D$, $\mathcal M$, $\mathcal N$ and family $\mathcal G$ of size~$3$ of forbidden induced subgraphs for polynomial-time \PBlistCOL{$H$} problems.}
\label{table}
\end{table}

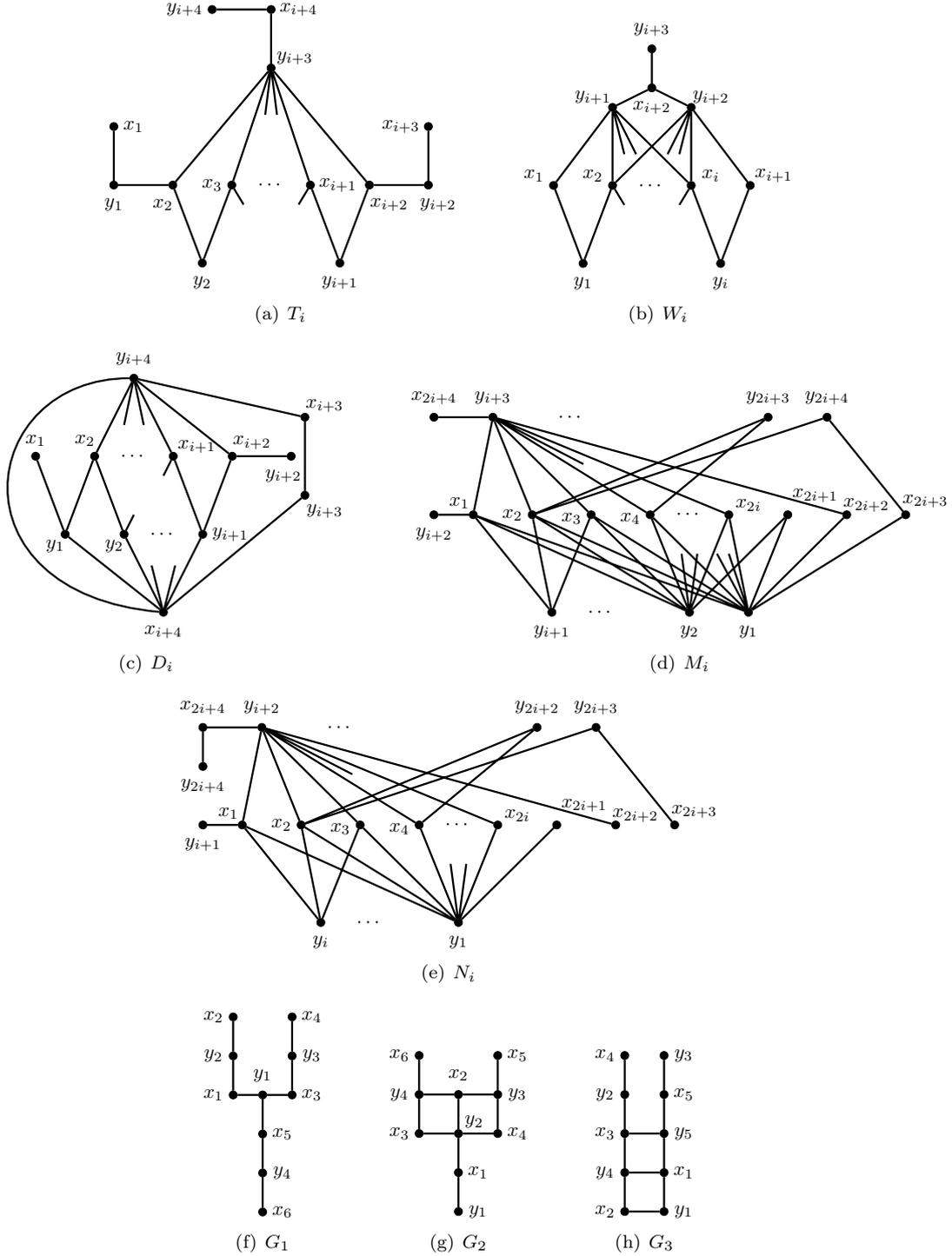
\begin{figure}[ht!]
\centering
\subfigure[$T_i$]{
\scalebox{0.85}{\begin{tikzpicture}[join=bevel,inner sep=0.5mm,scale=0.7]

\node[draw,shape=circle,fill](x1) at (2.5,1.5) {};
\path (x1)+(0.5,0) node {$x_1$};
\node[draw,shape=circle,fill](y1) at (2.5,0) {};
\path (y1)+(0,-0.5) node {$y_1$};
\node[draw,shape=circle,fill](x2) at (4,0) {};
\path (x2)+(-0.25,-0.5) node {$x_2$};
\node[draw,shape=circle,fill](y2) at (4.75,-2) {};
\path (y2)+(0,-0.5) node {$y_2$};
\node[draw,shape=circle,fill](x3) at (5.5,0) {};
\path (x3)+(-0.5,0) node {$x_3$};
\node[draw,shape=circle,fill](xi+1) at (7.5,0) {};
\path (xi+1)+(0.7,0) node {$x_{i+1}$};
\node[draw,shape=circle,fill](yi+1) at (8.25,-2) {};
\path (yi+1)+(0,-0.5) node {$y_{i+1}$};
\node[draw,shape=circle,fill](xi+2) at (9,0) {};
\path (xi+2)+(0.5,-0.5) node {$x_{i+2}$};
\node[draw,shape=circle,fill](yi+2) at (10.5,0) {};
\path (yi+2)+(0.25,-0.5) node {$y_{i+2}$};
\node[draw,shape=circle,fill](xi+3) at (10.5,1.5) {};
\path (xi+3)+(-0.7,0) node {$x_{i+3}$};
\node[draw,shape=circle,fill](yi+3) at (6.5,3) {};
\path (yi+3)+(0.6,0.25) node {$y_{i+3}$};
\node[draw,shape=circle,fill](xi+4) at (6.5,4.5) {};
\path (xi+4)+(0.7,0) node {$x_{i+4}$};
\node[draw,shape=circle,fill](yi+4) at (5,4.5) {};
\path (yi+4)+(-0.7,0) node {$y_{i+4}$};

\node at (6.5,0) {$\cdots$};

\draw[line width=1pt] (x1)--(y1)--(x2)--(y2)--(x3)--(yi+3)--(xi+1)--(yi+1)--(xi+2)--(yi+2)--(xi+3)
                      (x2)--(yi+3)--(xi+2) (yi+3)--(xi+4)--(yi+4)
                      (x3) -- ++(0.3,-0.5) (xi+1) -- ++(-0.3,-0.5)
                      (yi+3) -- ++(0.15,-1.2) (yi+3) -- ++(-0.15,-1.2);
\end{tikzpicture}}
}\qquad
\subfigure[$W_i$]{
\scalebox{0.85}{\begin{tikzpicture}[join=bevel,inner sep=0.5mm,scale=0.7]

\node[draw,shape=circle,fill](x1) at (0,0) {};
\path (x1)+(-0.5,0.25) node {$x_1$};
\node[draw,shape=circle,fill](x2) at (1.5,0) {};
\path (x2)+(-0.5,0.25) node {$x_2$};
\node[draw,shape=circle,fill](y1) at (0.75,-2) {};
\path (y1)+(0,-0.5) node {$y_1$};
\node[draw,shape=circle,fill](xi) at (3.5,0) {};
\path (xi)+(0.5,0.25) node {$x_i$};
\node[draw,shape=circle,fill](yi) at (4.25,-2) {};
\path (yi)+(0,-0.5) node {$y_i$};
\node[draw,shape=circle,fill](xi+1) at (5,0) {};
\path (xi+1)+(0.6,0.25) node {$x_{i+1}$};
\node[draw,shape=circle,fill](yi+1) at (1.5,2) {};
\path (yi+1)+(-0.5,0.25) node {$y_{i+1}$};
\node[draw,shape=circle,fill](yi+2) at (3.5,2) {};
\path (yi+2)+(0.5,0.25) node {$y_{i+2}$};
\node[draw,shape=circle,fill](xi+2) at (2.5,2.5) {};
\path (xi+2)+(0,-0.5) node {$x_{i+2}$};
\node[draw,shape=circle,fill](yi+3) at (2.5,3.5) {};
\path (yi+3)+(0,0.5) node {$y_{i+3}$};

\node at (2.5,0) {$\cdots$};

\draw[line width=1pt] (x1)--(y1)--(x2)--(yi+1)--(xi)--(yi)--(xi+1)--(yi+2)--(x2)
                      (x1)--(yi+1)--(xi+2)--(yi+3) (xi+2)--(yi+2)--(xi)
                      (x2) -- ++(0.3,-0.5) (xi) -- ++(-0.3,-0.5)
                      (yi+1) -- ++(0.3,-1.2) (yi+1) -- ++(0.6,-1.2)
                      (yi+2) -- ++(-0.3,-1.2) (yi+2) -- ++(-0.6,-1.2);

\end{tikzpicture}}
}\qquad
\subfigure[$D_i$]{
\scalebox{0.85}{\begin{tikzpicture}[join=bevel,inner sep=0.5mm,scale=0.7]

\node[draw,shape=circle,fill](x1) at (0,0) {};
\path (x1)+(0,0.4) node {$x_1$};
\node[draw,shape=circle,fill](x2) at (1.5,0) {};
\path (x2)+(-0.25,0.4) node {$x_2$};
\node[draw,shape=circle,fill](y1) at (0.75,-2) {};
\path (y1)+(-0.25,-0.3) node {$y_1$};
\node[draw,shape=circle,fill](y2) at (2.25,-2) {};
\path (y2)+(-0.25,-0.3) node {$y_2$};

\node[draw,shape=circle,fill](xi+1) at (3.5,0) {};
\path (xi+1)+(0.6,0.25) node {$x_{i+1}$};
\node[draw,shape=circle,fill](yi+1) at (4.25,-2) {};
\path (yi+1)+(0.7,0) node {$y_{i+1}$};
\node[draw,shape=circle,fill](xi+2) at (5,0) {};
\path (xi+2)+(0.5,0.35) node {$x_{i+2}$};
\node[draw,shape=circle,fill](yi+2) at (6.5,0) {};
\path (yi+2)+(-0.2,-0.5) node {$y_{i+2}$};

\node[draw,shape=circle,fill](xi+3) at (6.85,1) {};
\path (xi+3)+(0.5,0.3) node {$x_{i+3}$};
\node[draw,shape=circle,fill](yi+3) at (6.85,-1) {};
\path (yi+3)+(0.5,-0.4) node {$y_{i+3}$};

\node[draw,shape=circle,fill](xi+4) at (3.25,-4) {};
\path (xi+4)+(0,-0.5) node {$x_{i+4}$};
\node[draw,shape=circle,fill](yi+4) at (2.5,2) {};
\path (yi+4)+(0,0.5) node {$y_{i+4}$};

\node at (2.5,0) {$\cdots$};
\node at (3.25,-2) {$\cdots$};

\draw[line width=1pt] (x1)--(y1)--(xi+4)--(y2)--(x2)--(y1)
                      (x2)--(yi+4)--(xi+1)--(yi+1)
                      (xi+4)--(yi+1)--(xi+2)--(yi+2)
                      (xi+2)--(yi+4) .. controls +(-4,0) and +(-5.5,0.5) .. (xi+4)--(yi+3)--(xi+3)--(yi+4)
                      (y2) -- ++(0.25,0.5) (xi+1) -- ++(-0.25,-0.5)
                      (yi+4) -- ++(0.25,-1.2) (yi+4) -- ++(-0.25,-1.2)
                      (xi+4) -- ++(0.3,1.2) (xi+4) -- ++(-0.3,1.2);

\end{tikzpicture}}
}\qquad
\subfigure[$M_i$]{
\scalebox{0.85}{\begin{tikzpicture}[join=bevel,inner sep=0.5mm,scale=0.7]

\node[draw,shape=circle,fill](x1) at (0,0) {};
\path (x1)+(-0.35,0.3) node {$x_1$};
\node[draw,shape=circle,fill](x2) at (1.5,0) {};
\path (x2)+(-0.5,0) node {$x_2$};
\node[draw,shape=circle,fill](x3) at (3,0) {};
\path (x3)+(-0.5,-0.1) node {$x_3$};
\node[draw,shape=circle,fill](x4) at (4.5,0) {};
\path (x4)+(-0.5,-0.1) node {$x_4$};

\node[draw,shape=circle,fill](x2i) at (6.5,0) {};
\path (x2i)+(0.5,0.25) node {$x_{2i}$};
\node[draw,shape=circle,fill](x2i+1) at (8,0) {};
\path (x2i+1)+(0.7,0.5) node {$x_{2i+1}$};
\node[draw,shape=circle,fill](x2i+2) at (9.5,0) {};
\path (x2i+2)+(0.5,0.25) node {$x_{2i+2}$};
\node[draw,shape=circle,fill](x2i+3) at (11,0) {};
\path (x2i+3)+(0.5,0.35) node {$x_{2i+3}$};
\node[draw,shape=circle,fill](x2i+4) at (-1,2.5) {};
\path (x2i+4)+(0,0.5) node {$x_{2i+4}$};

\node[draw,shape=circle,fill](y1) at (7,-2.5) {};
\path (y1)+(0,-0.5) node {$y_1$};
\node[draw,shape=circle,fill](y2) at (5.5,-2.5) {};
\path (y2)+(0,-0.5) node {$y_2$};
\node[draw,shape=circle,fill](yi+1) at (2,-2.5) {};
\path (yi+1)+(0,-0.5) node {$y_{i+1}$};
\node[draw,shape=circle,fill](yi+2) at (-1,0) {};
\path (yi+2)+(0,-0.5) node {$y_{i+2}$};
\node[draw,shape=circle,fill](yi+3) at (0.5,2.5) {};
\path (yi+3)+(0,0.5) node {$y_{i+3}$};

\node[draw,shape=circle,fill](y2i+3) at (7.5,2.5) {};
\path (y2i+3)+(0,0.5) node {$y_{2i+3}$};
\node[draw,shape=circle,fill](y2i+4) at (9,2.5) {};
\path (y2i+4)+(0,0.5) node {$y_{2i+4}$};

\node at (5.5,0) {$\cdots$};
\node at (3.25,-2.5) {$\cdots$};
\node at (2.5,2.5) {$\cdots$};

\draw[line width=1pt] (yi+2)--(x1)--(yi+1)--(x2)--(y2)--(x1)--(y1)--(x3)--(yi+1)
                      (x3)--(y2)--(x2i)--(yi+3)--(x2i+2)--(y1)--(x2i+1)--(y2)
                      (y2)--(x4)--(y1)--(x2i+3)--(y2i+4)--(x2)--(yi+3)--(x1)
                      (x2i+4)--(yi+3)--(x3) (x2)--(y1)--(x2i)
                      (yi+3)--(x4)--(y2i+3)--(x2)
                      (y1) -- ++(-0.5,1.5) (y1) -- ++(-0.8,1.5)
                      (y2) -- ++(-0.2,1.5) (y2) -- ++(0.2,1.5)
                      (yi+3) -- ++(2.3,-1.2) ;

\end{tikzpicture}}
}\qquad
\subfigure[$N_i$]{
\scalebox{0.85}{\begin{tikzpicture}[join=bevel,inner sep=0.5mm,scale=0.7]

\node[draw,shape=circle,fill](x1) at (0,0) {};
\path (x1)+(-0.35,0.3) node {$x_1$};
\node[draw,shape=circle,fill](x2) at (1.5,0) {};
\path (x2)+(-0.5,0) node {$x_2$};
\node[draw,shape=circle,fill](x3) at (3,0) {};
\path (x3)+(-0.5,-0.1) node {$x_3$};
\node[draw,shape=circle,fill](x4) at (4.5,0) {};
\path (x4)+(-0.5,-0.1) node {$x_4$};

\node[draw,shape=circle,fill](x2i) at (6.5,0) {};
\path (x2i)+(0.5,0.25) node {$x_{2i}$};
\node[draw,shape=circle,fill](x2i+1) at (8,0) {};
\path (x2i+1)+(0.7,0.5) node {$x_{2i+1}$};
\node[draw,shape=circle,fill](x2i+2) at (9.5,0) {};
\path (x2i+2)+(0.5,0.25) node {$x_{2i+2}$};
\node[draw,shape=circle,fill](x2i+3) at (11,0) {};
\path (x2i+3)+(0.5,0.35) node {$x_{2i+3}$};
\node[draw,shape=circle,fill](x2i+4) at (-1,2.5) {};
\path (x2i+4)+(0,0.5) node {$x_{2i+4}$};

\node[draw,shape=circle,fill](y1) at (5.5,-2.5) {};
\path (y1)+(0,-0.5) node {$y_1$};
\node[draw,shape=circle,fill](yi) at (2,-2.5) {};
\path (yi)+(0,-0.5) node {$y_{i}$};
\node[draw,shape=circle,fill](yi+1) at (-1,0) {};
\path (yi+1)+(0,-0.5) node {$y_{i+1}$};
\node[draw,shape=circle,fill](yi+2) at (0.5,2.5) {};
\path (yi+2)+(0,0.5) node {$y_{i+2}$};

\node[draw,shape=circle,fill](y2i+2) at (7.5,2.5) {};
\path (y2i+2)+(0,0.5) node {$y_{2i+2}$};
\node[draw,shape=circle,fill](y2i+3) at (9,2.5) {};
\path (y2i+3)+(0,0.5) node {$y_{2i+3}$};

\node[draw,shape=circle,fill](y2i+4) at (-1,1.5) {};
\path (y2i+4)+(0,-0.5) node {$y_{2i+4}$};

\node at (5.5,0) {$\cdots$};
\node at (3.25,-2.5) {$\cdots$};
\node at (2.5,2.5) {$\cdots$};

\draw[line width=1pt] (yi+1)--(x1)--(yi)--(x2)--(y1)--(x1) (x3)--(yi)
                      (x3)--(y1)--(x2i)--(yi+2)--(x2i+2) (x2i+1)--(y1)
                      (y1)--(x4) (x2i+3)--(y2i+3)--(x2)--(yi+2)--(x1)
                      (y2i+4)--(x2i+4)--(yi+2)--(x3)
                      (yi+2)--(x4)--(y2i+2)--(x2)
                      (y1) -- ++(-0.2,1.5) (y1) -- ++(0.2,1.5)
                      (yi+2) -- ++(2.3,-1.2);

\end{tikzpicture}}
}\\
\subfigure[$G_1$]{
\scalebox{0.85}{\begin{tikzpicture}[join=bevel,inner sep=0.5mm,scale=0.7]

\node[draw,shape=circle,fill](x1) at (0.25,0) {};
\path (x1)+(-0.5,0) node {$x_1$};
\node[draw,shape=circle,fill](y2) at (0.25,1) {};
\path (y2)+(-0.5,0) node {$y_2$};
\node[draw,shape=circle,fill](x2) at (0.25,2) {};
\path (x2)+(-0.5,0) node {$x_2$};

\node[draw,shape=circle,fill](x3) at (1.75,0) {};
\path (x3)+(0.5,0) node {$x_3$};
\node[draw,shape=circle,fill](y3) at (1.75,1) {};
\path (y3)+(0.5,0) node {$y_3$};
\node[draw,shape=circle,fill](x4) at (1.75,2) {};
\path (x4)+(0.5,0) node {$x_4$};

\node[draw,shape=circle,fill](y1) at (1,0) {};
\path (y1)+(0,0.5) node {$y_1$};
\node[draw,shape=circle,fill](x5) at (1,-1) {};
\path (x5)+(0.5,0) node {$x_5$};
\node[draw,shape=circle,fill](y4) at (1,-2) {};
\path (y4)+(0.5,0) node {$y_4$};
\node[draw,shape=circle,fill](x6) at (1,-3) {};
\path (x6)+(0.5,0) node {$x_6$};

\draw[line width=1pt] (x2)--(y2)--(x1)--(y1)--(x5)--(y4)--(x6)
                      (y1)--(x3)--(y3)--(x4);

\end{tikzpicture}}
}\qquad
\subfigure[$G_2$]{
\scalebox{0.85}{\begin{tikzpicture}[join=bevel,inner sep=0.5mm,scale=0.7]

\node[draw,shape=circle,fill](y4) at (0,0) {};
\path (y4)+(-0.5,0) node {$y_4$};
\node[draw,shape=circle,fill](x6) at (0,1) {};
\path (x6)+(-0.5,0) node {$x_6$};
\node[draw,shape=circle,fill](x3) at (0,-1) {};
\path (x3)+(-0.5,0) node {$x_3$};

\node[draw,shape=circle,fill](y3) at (2,0) {};
\path (y3)+(0.5,0) node {$y_3$};
\node[draw,shape=circle,fill](x5) at (2,1) {};
\path (x5)+(0.5,0) node {$x_5$};
\node[draw,shape=circle,fill](x4) at (2,-1) {};
\path (x4)+(0.5,0) node {$x_4$};

\node[draw,shape=circle,fill](x2) at (1,0) {};
\path (x2)+(0,0.5) node {$x_2$};
\node[draw,shape=circle,fill](y2) at (1,-1) {};
\path (y2)+(0.4,0.3) node {$y_2$};
\node[draw,shape=circle,fill](x1) at (1,-2) {};
\path (x1)+(0.5,0) node {$x_1$};
\node[draw,shape=circle,fill](y1) at (1,-3) {};
\path (y1)+(0.5,0) node {$y_1$};

\draw[line width=1pt] (x6)--(y4)--(x2)--(y3)--(x5)
                      (y4)--(x3)--(y2)--(x4)--(y3)
                      (x2)--(y2)--(x1)--(y1);

\end{tikzpicture}}
}\qquad
\subfigure[$G_3$]{
\scalebox{0.85}{\begin{tikzpicture}[join=bevel,inner sep=0.5mm,scale=0.7]

\node[draw,shape=circle,fill](x2) at (0,0) {};
\path (x2)+(-0.5,0) node {$x_2$};
\node[draw,shape=circle,fill](y4) at (0,1) {};
\path (y4)+(-0.5,0) node {$y_4$};
\node[draw,shape=circle,fill](x3) at (0,2) {};
\path (x3)+(-0.5,0) node {$x_3$};
\node[draw,shape=circle,fill](y2) at (0,3) {};
\path (y2)+(-0.5,0) node {$y_2$};
\node[draw,shape=circle,fill](x4) at (0,4) {};
\path (x4)+(-0.5,0) node {$x_4$};

\node[draw,shape=circle,fill](y1) at (1,0) {};
\path (y1)+(0.5,0) node {$y_1$};
\node[draw,shape=circle,fill](x1) at (1,1) {};
\path (x1)+(0.5,0) node {$x_1$};
\node[draw,shape=circle,fill](y5) at (1,2) {};
\path (y5)+(0.5,0) node {$y_5$};
\node[draw,shape=circle,fill](x5) at (1,3) {};
\path (x5)+(0.5,0) node {$x_5$};
\node[draw,shape=circle,fill](y3) at (1,4) {};
\path (y3)+(0.5,0) node {$y_3$};

\draw[line width=1pt] (x4)--(y2)--(x3)--(y5)--(x5)--(y3)
                      (x3)--(y4)--(x1)--(y5)
                      (y4)--(x2)--(y1)--(x1);
\end{tikzpicture}}
}

\caption{Illustration of the families defined in Table~\ref{table} (except the cycles in $\mathcal C$).}
\label{fig:table}
\end{figure}

}

\shortpaper{In fact, one can show the following result (see~\cite{fullversion} for a proof).

\begin{theorem}\label{thm:table}
Let $H$ be a graph in the characterization of forbidden induced subgraphs of~\cite{TM76} that is no an even cycle. Then, \PBtropCOL{$H$} is polynomial-time solvable.
\end{theorem}
}

In this section, we first turn our attention to the family of even cycles of length at least~$6$. We show that \PBtropCOL{$C_{2k}$} is polynomial-time solvable for any $k\leq 6$. On the other hand, for any $k\geq 24$, \PBtropCOL{$C_{2k}$} is NP-complete. 
\longpaper{We then prove that for all other minimal graphs $H$ described in Table~\ref{table}, \PBtropCOL{$H$} is polynomial-time solvable.}

\subsection{Polynomial-time cases for even cycles}

We now prove that the tropical homomorphism problems for small even cycles are polynomial-time solvable.

\begin{theorem} \label{thm: C12}
For each integer $k$ with $2\leq k\leq 6$, \PBtropCOL{$C_{2k}$} is polynomial-time solvable.
\end{theorem}
\begin{proof}
Since \PBlistCOL{$C_4$} is polynomial-time solvable, \PBtropCOL{$C_{4}$} is polynomial-time solvable.

We will consider all cases $k\in\{3,4,5,6\}$ separately. But in each of those cases we note that 
if $(C_{2k}, c)$ is not a core, then the core is path, and since the \PBlistCOL{$P_k$} is polynomial-time solvable for any $k\geq 1$, \PBtropCOL{$C_{2k}$} would also be polynomial-time solvable. Hence in the rest of the proof we always assume $(C_{2k}, c)$ is a core. Furthermore, by Proposition~\ref{prop:bipartite-hom}, we can assume that the colour sets of $c$ in $X$ and $Y$ are disjoint.

\medskip

First, assume $k=3$. There are three vertices in each part of the bipartition of $C_6$. If one vertex is coloured with a colour not present anywhere else in the part, Lemma~\ref{lemm:unique-feature} implies again that \PBCOL{$(C_6,c)$} is polynomial-time solvable. Hence, we can assume that each part of the bipartition is monochromatic. But then $(C_6,c)$ is not a core, a contradiction with our assumption.

\medskip

Suppose $k=4$. There are four vertices in each part of the bipartition $(X,Y)$ of $C_8$. If there is a vertex that, in $c$, is the only one coloured with its colour, since \PBlistCOL{$P_k$} is polynomial-time solvable for any $k\geq 1$, by Lemma~\ref{lemm:unique-feature} \PBCOL{$(C_8,c)$} is polynomial-time solvable. Hence we may assume that each colour appears at least twice, in particular each part of the bipartition is coloured with either one or two colours. If some part, say $X$, is coloured with only one colour (say Blue) then $(C_8,c)$ is not a core which again contradicts our assumption.
Hence, in each part, there are exactly two vertices of each colour. In this case we can use Lemma~\ref{lemm:2SAT} with $S_1$, $S_2$, $S_3$ and $S_4$ being the four sets of two vertices with the same colour. It follows that \PBCOL{$(C_8,c)$} is polynomial-time solvable.



\medskip

Assume that $k=5$, and let $c$ be a vertex-colouring of $C_{10}$. By similar arguments as in the proof of Theorems~\ref{thm:small_graph} and~\ref{thm:SmallTree}, using Lemma~\ref{lemm:unique-feature} and the fact that $(H,c)$ should not be homomorphic to a $P_2$- or $P_3$-subgraph, each part of the bipartition $(X,Y)$ contains exactly two vertices of one colour and three vertices of another colour, say $X$ has three vertices coloured~$1$ and two vertices coloured~$2$, and $Y$ has three vertices coloured~$a$ and two vertices coloured~$b$.

The cyclic order of the colours of $X$ can be either $1-1-1-2-2$ or $1-1-2-1-2$ (up to permutation of colours and other symmetries). If this order is $1-1-1-2-2$, then the vertex of $Y$ adjacent to the two vertices coloured~$2$ satisfies the hypothesis of Lemma~\ref{lemm:unique-feature} and hence \PBCOL{$(C_{10},c)$} is polynomial-time solvable. The same argument can be applied to $Y$, hence the cyclic order of the colours of $Y$ is $a-a-b-a-b$.

Hence, there is a unique vertex $y$ of $Y$ whose two neighbours are coloured~$1$. If $c(y)=b$, then the second vertex of $Y$ coloured~$b$ is in the centre of a $3$-vertex path coloured $1-b-2$ that satisfies the hypothesis of Lemma~\ref{lemm:unique-feature}, hence \PBCOL{$(C_{10},c)$} is polynomial-time solvable. Therefore, we have $c(y)=a$. By the same argument, the unique vertex of $X$ adjacent to two vertices of $Y$ coloured~$a$ must be coloured~$1$. Therefore, up to symmetries $c$ is one of the three colourings $1-a-1-a-2-b-1-a-2-b$, $1-a-1-b-2-a-1-a-2-b$ and $1-a-1-b-2-a-1-b-2-a$ (in the cyclic order).

We are going to use the Lemma~\ref{lemm:2SAT} to conclude the case $k=5$. In a homomorphism to $(C_{10},c)$, a vertex coloured~$2$ or~$b$ can only be mapped to the two vertices in $(C_{10},c)$ of the corresponding colour. A vertex $v$ coloured~$1$ adjacent to at least one vertex coloured~$b$ or a vertex coloured~$a$ adjacent to at least one vertex coloured~$2$ also can only be mapped to two vertices of $(C_{10},c)$ (the ones having the same properties as $v$). However, a vertex coloured~$1$ all whose neighbours are coloured~$a$ can be mapped to three different vertices in $(C_{10},c)$ (say $x_1$, $x_2$, $x_3$, the vertices coloured~$1$, that all have a neighbour coloured~$a$). But at least one of $x_1$, $x_2$, $x_3$, say $x_1$, has a common neighbour coloured~$a$ with one of the two other vertices (say $x_2$). Therefore, if there is a homomorphism $h$ of some tropical graph $(G,c_1)$ to $(C_{10},c)$ mapping a vertex $v$ of $G$ coloured~$1$ all whose neighbours are coloured~$a$ to $x_1$, we can modify $h$ so that $v$ is mapped to $x_2$ instead. In other words, there is a homomorphism of $(G,c_1)$ to $(C_{10},c)$ where none of the vertices coloured~$1$ all whose neighbours are coloured~$a$ is mapped to $x_1$. Therefore such vertices have two possible targets: $x_2$ and $x_3$. The same is true for vertices coloured~$a$ all whose neighbours are coloured~$1$. Thus, $(C_{10},c)$ satisfies the hypothesis of Lemma~\ref{lemm:2SAT} and \PBCOL{$(C_{10},c)$} is polynomial-time solvable.

\medskip

Finally, assume now that $k=6$. Again, using Lemma~\ref{lemm:unique-feature}, we can assume than each part of the bipartition has at most three colours, and each colour appears at least twice. Furthermore, if there are exactly three colours in each part, each colour appears exactly twice and hence \PBCOL{$(C_{12},c)$} is polynomial-time solvable by Lemma~\ref{lemm:2SAT}. If one part of the bipartition has one colour and the other has at most two colours, then $(C_{12},c)$ would not be a core. Therefore, the numbers of colours of the parts in the bipartition are either one and three, two and three, or two and two.

Assume that one part, say $X$, is monochromatic (say Red) and the other, $Y$, has three colours (thus two vertices of each colour). For the graph to be a core and not satisfy Lemma~\ref{lemm:unique-feature}, the three colours of $Y$ must form the cyclic pattern $x-y-z-x-y-z$. In this case, considering any vertex $v$ of colour Red in an input tropical graph $(G,c_1)$, in any homomorphism $(G,c_1)\to (C_{12},c)$, all the neighbours of $v$ with the same colour must be identified. Furthermore, no Red vertex in $(G,c_1)$ can have neighbours of three distinct colours. Therefore, the mapping of each connected component is forced after making a choice for one vertex. Since there are two choices per vertex, we have a polynomial-time algorithm for \PBCOL{$(C_{12},c)$}.

Assume now that one part, say $X$, contains two colours ($a$ and $b$) and the other, $Y$, contains three colours ($x$, $y$ and $z$). Note that there are exactly two vertices of each colour in $Y$. We are going to use Lemma~\ref{lemm:2SAT} to conclude this case. A vertex of some input graph $(G,c_1)$ coloured $x$, $y$ or $z$ can only be mapped to two possible vertices in $(C_{12},c)$. A vertex of $(G,c_1)$ coloured $a$ or $b$ (say $a$) and having all its neighbours of the same colour, say $x$, might be mapped to more than two vertices of $(C_{12},c)$. However, once again, there are always two of these vertices that, together, are adjacent to all the vertices of colour~$x$ (indeed, there are only two vertices of colour~$x$). These two vertices are the designated targets for Lemma~\ref{lemm:2SAT}. A vertex coloured~$a$ (or~$b$) with two different colours in its neighbourhood can only be mapped to two possible vertices if there is no pattern $x-a-y-a-x-a-y$ in the graph (up to permutation of colours). Hence, if there is no such pattern in the graph (up to permutation of colours), $(C_{12},c)$ satisfies the hypothesis of Lemma~\ref{lemm:2SAT} and \PBCOL{$(C_{12},c)$} is polynomial-time solvable. On the other hand, if there is a pattern $x-a-y-a-x-a-y$ in the graph, then there is a unique path coloured $a-x-b$ or $a-y-b$ in the graph and, by Lemma~\ref{lemm:unique-feature}, \PBCOL{$(C_{12},c)$} is polynomial-time solvable as well.

Therefore, we are left to consider the cases where there are exactly two colours in each part. We assume first that there are two vertices coloured~$a$ and four vertices coloured~$b$ in one part, say $X$. If the neighbours of vertices of colour~$a$ all have the same colour, say~$x$, then $(C_{12},c)$ is not a core because it can be mapped to its sub-path coloured $a-x-b-y$. We suppose without loss of generality that the coloured cycle contains a path coloured $y-a-x-b$. Then, if there is no other path coloured $y-a-x-b$, by Lemma~\ref{lemm:unique-feature} \PBCOL{$(C_{12},c)$} is polynomial-time solvable. Therefore, there is another such path in $(C_{12},c)$. If this other path is part of a path $x-a-y-a-x$, then the problem is polynomial-time solvable by applying Lemma~\ref{lemm:unique-feature} to the star $a-y-a$. Up to symmetry, we are left with two cases: $y-a-x-b-.-b-.-a-.-b-.-b$ or $y-a-x-b-.-b-.-b-.-a-.-b$ (where a dot could be colour $x$ or $y$). The first case must be $y-a-x-b-.-b-y-a-x-b-.-b$, because otherwise, $(C_{12},c)$ is not a core. Any placement of the remaining $x$'s and $y$'s yields a polynomial-time solvable case using Lemma~\ref{lemm:distinct-vertex-boundary}. Similarly, the second case must be $y-a-x-b-.-b-.-b-y-a-x-b$. Then, \PBCOL{$(C_{12},c)$} is polynomial-time solvable because of Lemma~\ref{lemm:distinct-edge-boundary}, with $a-x-b-y-a$ as forcing set and $x-b-.-b-.-b-y$, which contains no vertex coloured~$a$, as the other set.

Finally, we can assume, without loss of generality, that there are exactly three vertices for each of the two colours in each part. There are three possible configurations in each part: $a-a-a-b-b-b$, $a-a-b-b-a-b$ or $a-b-a-b-a-b$, up to permutations of colours. If one part of the bipartition is in the first configuration, then, either we have the pattern $a-x-b$ or $a-y-b$ that satisfies the hypothesis of Lemma~\ref{lemm:unique-feature}, or we have two paths $a-x-b$, in which case, there is a unique path $a-y-a$ or $b-y-b$ which satisfies the hypothesis of Lemma~\ref{lemm:unique-feature}. Suppose some part of the bipartition is in the second case. Then, if we have the pattern $a-x-a-.-b-x-b-.-a-.-b-.$, there is a unique path $a-x-b$ satisfying the hypothesis of Lemma~\ref{lemm:unique-feature}. Otherwise, we have the pattern $a-x-a-.-b-y-b-.-a-.-b-.$, in which case we can apply Lemma~\ref{lemm:2SAT} in a similar way as for $C_{10}$. Therefore, both parts of the bipartition must be in the third configuration. But then every vertex is a forcing vertex and we can apply Lemma~\ref{lemm:all-forcing}. This completes the proof.
\end{proof}

\subsection{NP-completeness results for even cycles}
We now show that \PBtropCOL{$C_{2k}$} is NP-complete whenever $k\geq 24$. 
We present a proof using a specific $4$-tropical $48$-cycle. The proof holds similarly for any larger even cycle. It also works similarly for some $3$-tropical cycles $C_{2k}$ for $k\geq 24$ and for $2$-tropical cycles $C_{2k}$ for $k\geq 27$.

We use the colour set $\{G, B, R, Y\}$ (for Green, Blue, Red and Yellow).

We define $P_{x,y}$ to be a tropical path of length~$8$, with vertices 
$x=x_0, x_1, \dots, x_7, x_8=y$ where $\{c(x),c(y)\}=\{G, B\}$, $c(x_5)=R$ and all others are coloured Yellow.
Thus, $P_{x,y}$ represents one of the two non-isomorphic tropical graphs from Figure~\ref{fig:Pxy}. 
The distance of the only vertex of colour $R$ from the two ends defines an orientation from one end to another. Thus, in our figures, an arc between two vertices $u$ and $v$ is a $P_{uv}$ path.

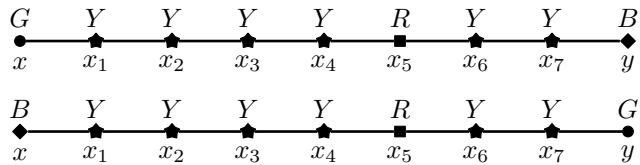
\begin{figure}[ht!]
\centering
\scalebox{1}{\begin{tikzpicture}[join=bevel,inner sep=0.5mm,scale=1]

\begin{scope}
\node[draw,shape=circle,color=black,fill](x0) at (0,0) {};
\path (x0)+(0,-0.3) node {$x$};
\path (x0)+(0,+0.3) node {$G$};
\node[draw,shape=star,color=black,fill](x1) at (1,0) {};
\path (x1)+(0,-0.3) node {$x_1$};
\path (x1)+(0,+0.3) node {$Y$};
\node[draw,shape=star,color=black,fill](x2) at (2,0) {};
\path (x2)+(0,-0.3) node {$x_2$};
\path (x2)+(0,+0.3) node {$Y$};
\node[draw,shape=star,color=black,fill](x3) at (3,0) {};
\path (x3)+(0,-0.3) node {$x_3$};
\path (x3)+(0,+0.3) node {$Y$};
\node[draw,shape=star,color=black,fill](x4) at (4,0) {};
\path (x4)+(0,-0.3) node {$x_4$};
\path (x4)+(0,+0.3) node {$Y$};
\node[draw,shape=rectangle,scale=1.4,color=black,fill](x5) at (5,0) {};
\path (x5)+(0,-0.3) node {$x_5$};
\path (x5)+(0,+0.3) node {$R$};
\node[draw,shape=star,color=black,fill](x6) at (6,0) {};
\path (x6)+(0,-0.3) node {$x_6$};
\path (x6)+(0,+0.3) node {$Y$};
\node[draw,shape=star,color=black,fill](x7) at (7,0) {};
\path (x7)+(0,-0.3) node {$x_7$};
\path (x7)+(0,+0.3) node {$Y$};
\node[draw,shape=diamond,color=black,fill](x8) at (8,0) {};
\path (x8)+(0,-0.3) node {$y$};
\path (x8)+(0,+0.3) node {$B$};

\draw[line width=1pt] (x0)--(x1)--(x2)--(x3)--(x4)--(x5)--(x6)--(x7)--(x8);
\end{scope}

\begin{scope}[yshift=-1.2cm]
\node[draw,shape=diamond,color=black,fill](x0) at (0,0) {};
\path (x0)+(0,-0.3) node {$x$};
\path (x0)+(0,+0.3) node {$B$};
\node[draw,shape=star,color=black,fill](x1) at (1,0) {};
\path (x1)+(0,-0.3) node {$x_1$};
\path (x1)+(0,+0.3) node {$Y$};
\node[draw,shape=star,color=black,fill](x2) at (2,0) {};
\path (x2)+(0,-0.3) node {$x_2$};
\path (x2)+(0,+0.3) node {$Y$};
\node[draw,shape=star,color=black,fill](x3) at (3,0) {};
\path (x3)+(0,-0.3) node {$x_3$};
\path (x3)+(0,+0.3) node {$Y$};
\node[draw,shape=star,color=black,fill](x4) at (4,0) {};
\path (x4)+(0,-0.3) node {$x_4$};
\path (x4)+(0,+0.3) node {$Y$};
\node[draw,shape=rectangle,scale=1.4,color=black,fill](x5) at (5,0) {};
\path (x5)+(0,-0.3) node {$x_5$};
\path (x5)+(0,+0.3) node {$R$};
\node[draw,shape=star,color=black,fill](x6) at (6,0) {};
\path (x6)+(0,-0.3) node {$x_6$};
\path (x6)+(0,+0.3) node {$Y$};
\node[draw,shape=star,color=black,fill](x7) at (7,0) {};
\path (x7)+(0,-0.3) node {$x_7$};
\path (x7)+(0,+0.3) node {$Y$};
\node[draw,shape=circle,color=black,fill](x8) at (8,0) {};
\path (x8)+(0,-0.3) node {$y$};
\path (x8)+(0,+0.3) node {$G$};

\draw[line width=1pt] (x0)--(x1)--(x2)--(x3)--(x4)--(x5)--(x6)--(x7)--(x8);
\end{scope}

%
%
\end{tikzpicture}}\caption{The two non-isomorphic graphs of type $P_{xy}$.}
\label{fig:Pxy}
\end{figure}

Similarly, $Q_{z,t}$ is defined to be a tropical path of length~$10$ with vertices 
$z=z_0, z_1, \dots, z_9, z_{10}=t$ where $\{c(z),c(t)\}=\{G, B\}$, $c(z_5)=R$ and all others are coloured Yellow. See Figure~\ref{fig:QZt} for an illustration.
In this case, as the only vertex of colour $R$ is at the same distance from both ends, the two possible colourings of the end-vertices correspond to isomorphic graphs. Hence, in our figures, a dotted edge will be used to represent a $Q$-type path between two vertices.

\begin{figure}[ht!]
\centering

\scalebox{1}{\begin{tikzpicture}[join=bevel,inner sep=0.5mm,scale=1]

\begin{scope}
\node[draw,shape=circle,color=black,fill](x0) at (0,0) {};
\path (x0)+(0,-0.3) node {$z$};
\path (x0)+(0,+0.3) node {$G$};
\node[draw,shape=star,color=black,fill](x1) at (1,0) {};
\path (x1)+(0,-0.3) node {$z_1$};
\path (x1)+(0,+0.3) node {$Y$};
\node[draw,shape=star,color=black,fill](x2) at (2,0) {};
\path (x2)+(0,-0.3) node {$z_2$};
\path (x2)+(0,+0.3) node {$Y$};
\node[draw,shape=star,color=black,fill](x3) at (3,0) {};
\path (x3)+(0,-0.3) node {$z_3$};
\path (x3)+(0,+0.3) node {$Y$};
\node[draw,shape=star,color=black,fill](x4) at (4,0) {};
\path (x4)+(0,-0.3) node {$z_4$};
\path (x4)+(0,+0.3) node {$Y$};
\node[draw,shape=rectangle,scale=1.4,color=black,fill](x5) at (5,0) {};
\path (x5)+(0,-0.3) node {$z_5$};
\path (x5)+(0,+0.3) node {$R$};
\node[draw,shape=star,color=black,fill](x6) at (6,0) {};
\path (x6)+(0,-0.3) node {$z_6$};
\path (x6)+(0,+0.3) node {$Y$};
\node[draw,shape=star,color=black,fill](x7) at (7,0) {};
\path (x7)+(0,-0.3) node {$z_7$};
\path (x7)+(0,+0.3) node {$Y$};
\node[draw,shape=star,color=black,fill](x8) at (8,0) {};
\path (x8)+(0,-0.3) node {$z_8$};
\path (x8)+(0,+0.3) node {$Y$};
\node[draw,shape=star,color=black,fill](x9) at (9,0) {};
\path (x9)+(0,-0.3) node {$z_9$};
\path (x9)+(0,+0.3) node {$Y$};
\node[draw,shape=diamond,color=black,fill](x10) at (10,0) {};
\path (x10)+(0,-0.3) node {$t$};
\path (x10)+(0,+0.3) node {$B$};

\draw[line width=1pt] (x0)--(x1)--(x2)--(x3)--(x4)--(x5)--(x6)--(x7)--(x8)--(x9)--(x10);
\end{scope}

%
\end{tikzpicture}}
\caption{The $Q$-type path $Q_{z,t}$.}
\label{fig:QZt}
\end{figure}
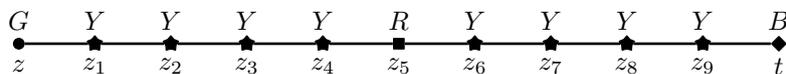

The following lemma is easy to observe.

\begin{lemma} \label{lem:arc_and_dotted_line}
The following is true.
\begin{enumerate}
\item $P_{x,y}$ admits a tropical homomorphism to $P_{u,v}$ if and only if $c(x)=c(u)$ and $c(y)=c(v)$.
\item $Q_{z,t}$ admits a tropical homomorphism to $P_{u,v}$ both in the case where $c(z)=c(u)$ and $c(t)=c(v)$, and in the case where $c(z)=c(v)$ and $c(t)=c(u)$.
\end{enumerate}
\end{lemma}

By Lemma~\ref{lem:arc_and_dotted_line}, in our abbreviated notation of arcs and dotted edges, a dotted edge can map to a dotted edge or to an arc as long as the colours of the end-vertices are preserved. However, to map an arc to another arc, not only the colours of the end-vertices must be preserved, but also the direction of the arc.

With our notation, the tropical directed $6$-cycle of Figure~\ref{fig:C48} corresponds to a $4$-tropical $48$-cycle, $(C_{48},c)$.

\begin{figure}[ht!]
\centering

\scalebox{1}{\begin{tikzpicture}[join=bevel,inner sep=0.5mm,scale=1]
\begin{scope}
\draw node[color=black,fill,circle,label=0*360/6-90:$g_0$,label=0*360/6+90:$G$](g0) at (0*360/6-90:1.3cm) {};
\draw node[color=black,fill,diamond,label=1*360/6-90:$b_0$,label=1*360/6+90:$B$](b0) at (1*360/6-90:1.3cm) {};
\draw node[color=black,fill,circle,label=2*360/6-90:$g_1$,label=2*360/6+90:$G$](g1) at (2*360/6-90:1.3cm) {};
\draw node[color=black,fill,diamond,label=3*360/6-90:$b_1$,label=3*360/6+90:$B$](b1) at (3*360/6-90:1.3cm) {};
\draw node[color=black,fill,circle,label=4*360/6-90:$g_2$,label=4*360/6+90:$G$](g2) at (4*360/6-90:1.3cm) {};
\draw node[color=black,fill,diamond,label=5*360/6-90:$b_2$,label=5*360/6+90:$B$](b2) at (5*360/6-90:1.3cm) {};

\draw[->, line width=1pt] (g0)--(b0);
\draw[->, line width=1pt] (b0)--(g1);
\draw[->, line width=1pt] (g1)--(b1);
\draw[->, line width=1pt] (b1)--(g2);
\draw[->, line width=1pt] (g2)--(b2);
\draw[->, line width=1pt] (b2)--(g0);
\end{scope}

\end{tikzpicture}}
\caption{A short representation of the $4$-tropical $48$-cycle $(C_{48},c)$.}
\label{fig:C48}
\end{figure}
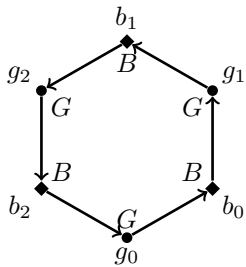

Our aim is to show that NAE $3$-SAT reduces (in polynomial time) to \PBCOL{$(C_{48},c)$}.

\begin{theorem} \label{thm:C48}
For any $k\geq 24$, \PBtropCOL{$C_{2k}$} is NP-complete.
\end{theorem}
\begin{proof}
We prove the statement when $k=24$ and observe that the same reduction holds for any $k\geq 24$. Indeed, one can make $P_{x,y}$ and $Q_{z,t}$ longer while still satisfying Lemma \ref{lem:arc_and_dotted_line}.

\PBCOL{$(C_{48},c)$} is clearly in NP. To show NP-hardness, we show that NAE $3$-SAT can be reduced in polynomial-time to \PBCOL{$(C_{48},c)$}. 

Let $(X,C)$ be an instance of NAE $3$-SAT. To partition $X$ into two parts, it is enough to decide, for each pair of elements of $X$, whether they are in a same part or not. Thus, we are expected to define a binary relation among variables which satisfies the following conditions.

\begin{enumerate}
\item $X_p\sim X_q \wedge X_q\sim X_r \Rightarrow X_q\sim X_r$ (Partition)
\item $X_p\nsim X_q \wedge X_q\nsim X_r \Rightarrow X_p\sim X_r$ (Partition into two parts)
\end{enumerate}

To build our gadget, we start with a partial gadget associated to each pair of variables of $X$. To each pair $x_i, x_j\in X$, we associate the $4$-tropical $6$-cycle $(C_{x_ix_j},c)$ of Figure~\ref{fig:gadgetCx1x2}. Here, $U_G$ (coloured Green) is a common vertex of all such cycles, but all other vertices are distinct.

\begin{figure}[ht!]
\centering

\scalebox{1}{\begin{tikzpicture}[join=bevel,inner sep=0.5mm,scale=1]
\begin{scope}
\draw node[color=black,fill,circle,label=0*360/6-90:$U_G$,label=0*360/6+90:$G$](g0) at (0*360/6-90:1.3cm) {};
\draw node[color=black,fill,diamond,label=1*360/6-90:$b^0_{x_ix_j}$,label=1*360/6+90:$B$](b0) at (1*360/6-90:1.3cm) {};
\draw node[color=black,fill,circle,label=2*360/6-90:$g^1_{x_ix_j}$,label=2*360/6+90:$G$](g1) at (2*360/6-90:1.3cm) {};
\draw node[color=black,fill,diamond,label=3*360/6-90:$b^1_{x_ix_j}$,label=3*360/6+90:$B$](b1) at (3*360/6-90:1.3cm) {};
\draw node[color=black,fill,circle,label=4*360/6-90:$g^2_{x_ix_j}$,label=4*360/6+90:$G$](g2) at (4*360/6-90:1.3cm) {};
\draw node[color=black,fill,diamond,label=5*360/6-90:$b^2_{x_ix_j}$,label=5*360/6+90:$B$](b2) at (5*360/6-90:1.3cm) {};

\draw[->, line width=1pt] (g0)--(b0);
\draw[->, line width=1pt] (b0)--(g1);
\draw[loosely dashed, line width=1pt] (g1)--(b1);
\draw[->, line width=1pt] (b1)--(g2);
\draw[loosely dashed, line width=1pt] (g2)--(b2);
\draw[loosely dashed, line width=1pt] (b2)--(g0);
\end{scope}

\end{tikzpicture}}
\caption{$(C_{x_ix_j},c)$}
\label{fig:gadgetCx1x2}
\end{figure}
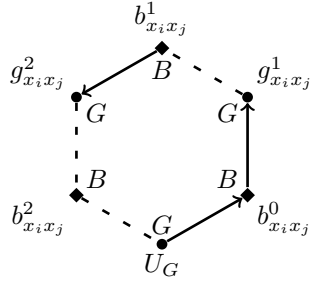

We are interested in possible mappings of this partial gadget into our tropical $48$-cycle, $(C_{48}, c)$ of Figure~\ref{fig:C48}. By the symmetries of $(C_{48},c)$, we assume, without loss of generality, that $U_G$ maps to $g_0$. Having this assumed, we observe the following crucial fact.

\begin{claim}\label{claim:C48}
There are exactly two possible homomorphisms of $(C_{x_ix_j},c)$ to $(C_{48},c)$.
\begin{enumerate}
\item A mapping $\sigma$ given by $\sigma(U_G)=g_0$, $\sigma(b^0_{x_ix_j})=b_0$, $\sigma(g^1_{x_ix_j})=g_1$, $\sigma(b^1_{x_ix_j})=b_1$, $\sigma(g^2_{x_ix_j})=g_2$ and $\sigma(b^2_{x_ix_j})=b_2$
\item A mapping $\rho$ give by $\rho(U_G)=g_0$, $\rho(b^0_{x_ix_j})=b_0$, $\rho(g^1_{x_ix_j})=g_1$, $\rho(b^1_{x_ix_j})=b_0$, 
$\rho(g^2_{x_ix_j})=g_1$ and $\rho(b^2_{x_ix_j})=b_0$
\end{enumerate}
\end{claim}

The main idea of our reduction lies in Claim~\ref{claim:C48}. After completing the description of our gadgets, we will have a $4$-tropical graph containing a copy of $C_{x_ix_j}$ for each pair $x_i,x_j$ of variables. If we find a homomorphism of this graph to $(C_{48},c)$, then its restriction to $C_{x_ix_j}$ is either a mapping of type $\sigma$, or of type $\rho$. A $\sigma$-mapping would correspond to assigning $x_i$ and $x_j$ to two different parts, and a $\rho$-mapping would correspond to assigning them to a same part of a partition of $X$.

\begin{observation}
It is never possible to map $b^2_{x_ix_j}$ to $b_1$ or to map $b^1_{x_ix_j}$ to $b_2$.
\end{observation}

To enforce the two conditions, partitioning $X$ into two parts by a binary relation, we add more structures. Consider the three partial gadgets $(C_{x_px_q},c)$,  $(C_{x_qx_r},c)$ and $(C_{x_px_r},c)$. 
Considering $b^1_{x_px_q}$ of $(C_{x_px_q},c)$, we choose vertices $b^2_{x_px_r}$ and $b^2_{x_qx_r}$ 
from $(C_{x_px_r},c)$ and  $(C_{x_qx_r},c)$ respectively, and connect them by a tree as in Figure~\ref{fig:treeCx1x2x3}.
The internal vertices of these trees are all new and distinct.

\begin{figure}[ht!]
\centering

\scalebox{0.75}{\begin{tikzpicture}[join=bevel,inner sep=0.5mm,scale=1]
\begin{scope}
\draw node[color=black,fill,diamond,label=0:$b^1_{x_px_q}$,label=180:$B$](b1x1x2) at (0,0) {};
\draw node[color=black,fill,diamond,label=-90:$b^2_{x_qx_r}$,label=90:$B$](b2x2x3) at (3,-1) {};
\draw node[color=black,fill,diamond,label=0:$b^2_{x_px_r}$,label=180:$B$](b2x1x3) at (0,-2) {};
\draw node[color=black,fill,circle,label=180:$G$](g0) at (0,-1) {};
\draw node[color=black,fill,circle,label=90:$G$](g1) at (2,-1) {};
\draw node[color=black,fill,diamond,label=90:$B$](b0) at (1,-1) {};

\draw[->, line width=1pt] (g1)--(b0);
\draw[->, line width=1pt] (b0)--(g0);
\draw[loosely dashed, line width=1pt] (b1x1x2)--(g0)--(b2x1x3);
\draw[loosely dashed, line width=1pt] (g1)--(b2x2x3);
\end{scope}
\end{tikzpicture}}
\caption{Tree connecting $b^1_{x_px_q}$, $b^2_{x_px_r}$ and $b^2_{x_qx_r}$}
\label{fig:treeCx1x2x3}
\end{figure}
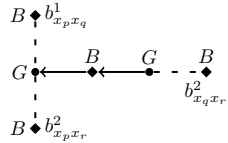

We build similar structures on $(b^1_{x_px_r},b^2_{x_qx_r},b^2_{x_px_q})$ and on $(b^1_{x_qx_r},b^2_{x_px_q},b^2_{x_px_r})$, where the order corresponds to the structure. Let $(C_{x_px_qx_r},c)$ be the resulting partial gadget (see Figure \ref{fig:Cx1x2x3}).

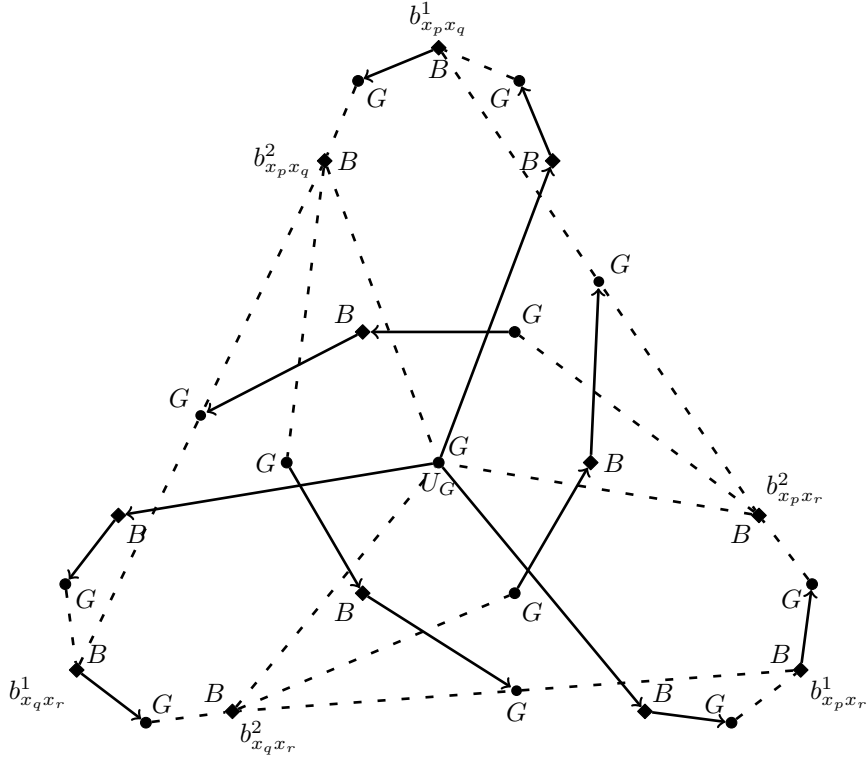
\begin{figure}[ht!]
\centering

\scalebox{1}{\begin{tikzpicture}[join=bevel,inner sep=0.5mm,scale=1]
\begin{scope}
\node[draw,circle,fill,label=-90:$U_G$,color=black,label=30:$G$](ug) at (0,0) {};
\begin{scope}[yshift=4cm]
\node[draw,color=black,fill,diamond,label=0*180/4+180:$B$](b0x12) at (0*180/4:1.5cm) {};
\node[draw,color=black,fill,circle,label=1*180/4+180:$G$](g1x12) at (1*180/4:1.5cm) {};
\node[draw,diamond,fill,label=2*180/4:$b^1_{x_px_q}$,color=black,label=2*180/4+180:$B$](b1x12) at (2*180/4:1.5cm) {};
\node[draw,color=black,fill,circle,label=3*180/4+180:$G$](g2x12) at (3*180/4:1.5cm) {};
\node[draw,diamond,fill,label=4*180/4:$b^2_{x_px_q}$,color=black,label=4*180/4+180:$B$](b2x12) at (4*180/4:1.5cm) {};
\draw[->, line width=1pt] (ug)--(b0x12);
\draw[->, line width=1pt] (b0x12)--(g1x12);
\draw[loosely dashed, line width=1pt] (g1x12)--(b1x12);
\draw[->, line width=1pt] (b1x12)--(g2x12);
\draw[loosely dashed, line width=1pt] (g2x12)--(b2x12);
\draw[loosely dashed, line width=1pt] (b2x12)--(ug);
\end{scope}

\begin{scope}[rotate=-120, yshift=4cm]
\node[draw,color=black,fill,diamond,label=0*180/4+60:$B$](b0x13) at (0*180/4:1.5cm) {};
\node[draw,color=black,fill,circle,label=1*180/4+60:$G$](g1x13) at (1*180/4:1.5cm) {};
\node[draw,diamond,fill,label=-120+2*180/4:$b^1_{x_px_r}$,color=black,label=2*180/4+60:$B$](b1x13) at (2*180/4:1.5cm) {};
\node[draw,color=black,fill,circle,label=3*180/4+60:$G$](g2x13) at (3*180/4:1.5cm) {};
\node[draw,diamond,fill,label=-120+4*180/4:$b^2_{x_px_r}$,color=black,label=4*180/4+60:$B$](b2x13) at (4*180/4:1.5cm) {};
\draw[->, line width=1pt] (ug)--(b0x13);
\draw[->, line width=1pt] (b0x13)--(g1x13);
\draw[loosely dashed, line width=1pt] (g1x13)--(b1x13);
\draw[->, line width=1pt] (b1x13)--(g2x13);
\draw[loosely dashed, line width=1pt] (g2x13)--(b2x13);
\draw[loosely dashed, line width=1pt] (b2x13)--(ug);
\end{scope}

\begin{scope}[rotate=120, yshift=4cm]
\node[draw,color=black,fill,diamond,label=0*180/4+300:$B$](b0x23) at (0*180/4:1.5cm) {};
\node[draw,color=black,fill,circle,label=1*180/4+300:$G$](g1x23) at (1*180/4:1.5cm) {};
\node[draw,diamond,fill,label=120+2*180/4:$b^1_{x_qx_r}$,color=black,label=2*180/4+300:$B$](b1x23) at (2*180/4:1.5cm) {};
\node[draw,color=black,fill,circle,label=3*180/4+300:$G$](g2x23) at (3*180/4:1.5cm) {};
\node[draw,diamond,fill,label=120+4*180/4:$b^2_{x_qx_r}$,color=black,label=4*180/4+300:$B$](b2x23) at (4*180/4:1.5cm) {};
\draw[->, line width=1pt] (ug)--(b0x23);
\draw[->, line width=1pt] (b0x23)--(g1x23);
\draw[loosely dashed, line width=1pt] (g1x23)--(b1x23);
\draw[->, line width=1pt] (b1x23)--(g2x23);
\draw[loosely dashed, line width=1pt] (g2x23)--(b2x23);
\draw[loosely dashed, line width=1pt] (b2x23)--(ug);
\end{scope}
\draw[loosely dashed, line width=1pt] (b1x12)--(b2x13) node[midway,color=black,fill,circle,label=30:$G$](g0x1){};
\draw[loosely dashed, line width=1pt] (b1x13)--(b2x23) node[midway,color=black,fill,circle,label=-90:$G$](g0x3){};
\draw[loosely dashed, line width=1pt] (b1x23)--(b2x12) node[midway,color=black,fill,circle,label=150:$G$](g0x2){};

\node[draw,color=black,fill,diamond,label=0*60:$B$](b0x1) at (0*60:2cm) {};
\node[draw,color=black,fill,circle,label=1*60:$G$](g1x2) at (1*60:2cm) {};
\node[draw,color=black,fill,diamond,label=2*60:$B$](b0x2) at (2*60:2cm) {};
\node[draw,color=black,fill,circle,label=3*60:$G$](g1x3) at (3*60:2cm) {};
\node[draw,color=black,fill,diamond,label=4*60:$B$](b0x3) at (4*60:2cm) {};
\node[draw,color=black,fill,circle,label=5*60:$G$](g1x1) at (5*60:2cm) {};

\draw[->, line width=1pt] (b0x1)--(g0x1);
\draw[->, line width=1pt] (g1x1)--(b0x1);
\draw[loosely dashed, line width=1pt] (b2x23)--(g1x1);

\draw[->, line width=1pt] (b0x2)--(g0x2);
\draw[->, line width=1pt] (g1x2)--(b0x2);
\draw[loosely dashed, line width=1pt] (b2x13)--(g1x2);

\draw[->, line width=1pt] (b0x3)--(g0x3);
\draw[->, line width=1pt] (g1x3)--(b0x3);
\draw[loosely dashed, line width=1pt] (b2x12)--(g1x3);
\end{scope}
\end{tikzpicture}}
\caption{$C_{x_px_qx_r}$}
\label{fig:Cx1x2x3}
\end{figure}

\begin{claim} \label{claim:nice_representation_of_equivalence}
In any mapping of $(C_{x_px_qx_r},c)$ to $(C_{48},c)$, an odd number of $(C_{x_ix_j},c)$ is mapped to $(C_{48},c)$ by a $\rho$-mapping. Furthermore, for any choice of an odd number of $(C_{x_ix_j},c)$ (that is either one or all three of them), there exists a mapping of $(C_{x_px_qx_r},c)$ to $(C_{48},c)$ which induces a $\rho$-mapping exactly on our choice.
\end{claim}
\emph{Proof of claim}
Indeed, each $(C_{x_ix_j},c)$ can be mapped to $(C_{48},c)$ only by $\sigma$ or $\rho$, which implies that there are eight ways to map the union of $(C_{x_px_q},c)$, $(C_{x_px_r},c)$ and $(C_{x_qx_r},c)$ to $(C_{48},c)$. Of these eight ways, four map an odd number of $(C_{x_ix_j},c)$ to $(C_{48},c)$ by a $\rho$-mapping. The four remaining ways are to map all $(C_{x_ix_j},c)$ to $(C_{48},c)$ by a $\sigma$-mapping, or to choose one of them to map by a $\sigma$-mapping and to map the two others by a $\rho$-mapping. One can check easily that the union of $(C_{x_px_q},c)$, $(C_{x_px_r},c)$, $(C_{x_qx_r},c)$ and the tree of Figure~\ref{fig:treeCx1x2x3} has six ways to be mapped to $(C_{48},c)$. Indeed, it is no longer possible to map all $(C_{x_ix_j},c)$ by $\sigma$ nor to map $(C_{x_px_r},c)$ by $\sigma$ and $(C_{x_px_q},c)$ and $(C_{x_qx_r},c)$ by $\rho$. By symmetry, this implies Claim~\ref{claim:nice_representation_of_equivalence}.~\smallqed

Finally, to complete the gadget, what remains is to forbid the possibility of a $\rho$-mapping for all three of $(C_{x_px_q},c)$, $(C_{x_px_r},c)$ and $(C_{x_qx_r},c)$ in the case where $(x_px_qx_r)$ is a clause in $C$. This is done by adding a $b^1_{x_px_q}b^2_{x_qx_r}$-path shown in Figure~\ref{fig:partialGadgetForClause}.

\begin{figure}[ht!]
\centering

\scalebox{1}{\begin{tikzpicture}[join=bevel,inner sep=0.5mm,scale=1]
\begin{scope}
\draw node[color=black,fill,diamond,label=-90:$b^1_{x_px_q}$,label=90:$B$](b1x1x2) at (0,0) {};
\draw node[color=black,fill,diamond,label=-90:$b^2_{x_qx_r}$,label=90:$B$](b2x2x3) at (4,0) {};
\draw node[color=black,fill,circle,label=90:$G$](x1) at (1,0) {};
\draw node[color=black,fill,diamond,label=90:$B$](x2) at (2,0) {};
\draw node[color=black,fill,circle,label=90:$G$](x3) at (3,0) {};

\draw[->, line width=1pt] (b1x1x2)--(x1);
\draw[->, line width=1pt] (x1)--(x2);
\draw[->, line width=1pt] (x2)--(x3);
\draw[loosely dashed, line width=1pt] (x3)--(b2x2x3);
\end{scope}

\end{tikzpicture}}
\caption{Partial clause gadget.}
\label{fig:partialGadgetForClause}
\end{figure}
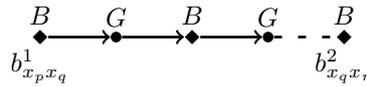

Let $f(X,C)$ the final gadget we have just built. Assuming that there are $v$ variables and $c$ clauses, the $4$-tropical graph $f(X,C)$ has $1+53\times v^2+132\times v^3+33\times c$ vertices. To complete our proof we want to prove the following.

$(X,C)$ is a YES instance of NAE $3$-SAT if and only if the $4$-tropical graph $f(X,C)$ admits a homomorphism to $(C_{48},c)$.

It follows directly form our construction that if $f(X,C) \to (C_{48},c)$, then $(X,C)$ is a YES instance of \textsc{NAE $3$-SAT}.
We need to show that if $(X,C)$ is a YES instance, then there exists a homomorphism of $f(X,C)$ to $(C_{48},c)$.

Let $(X,C)$ be a YES instance of \textsc{NAE $3$-SAT}. There exists a partition $p: X \rightarrow \{A,B\}$ such that every clause in $C$ is not fully included in $A$ or $B$. We build a homomorphism of $f(X,C)$ to $(C_{48},c)$ in the following way. $U_G$ is mapped to $g_0$. For each pair of variables $x_i, x_j \in X$, we map $C_{x_ix_j}$ by a $\rho$-mapping if and only if $p(x_i)=p(x_j)$, and by a $\sigma$-mapping otherwise. For every triple of variable $x_p, x_q, x_r \in X$, there is an odd number of pairs $x_i,x_j$ of variables in $\{x_p, x_q, x_r\}$ such that $p(x_i)=p(x_j)$. It follows from Claim \ref{claim:nice_representation_of_equivalence} that one can extend the mapping to any $C_{x_px_qx_r}$. Moreover, as two such structures only intersect on $C_{x_ix_j}$, we can extend the mapping to every $C_{x_px_qx_r}$. It only remains to map the $b^1_{x_px_q}b^2_{x_qx_r}$-path added for the clause, shown in Figure~\ref{fig:partialGadgetForClause}. If $(x_p, x_q, x_r)$ is a clause in $C$, then $p(x_p) \neq p(x_q)$ or $p(x_q) \neq p(x_r)$. It follows that $C_{x_px_q}$ or $C_{x_qx_r}$ is mapped by a $\sigma$-mapping, in which case the $b^1_{x_px_q}b^2_{x_qx_r}$-path shown in Figure~\ref{fig:partialGadgetForClause} can also be mapped. We have shown that there is a homomorphism of $f(X,C)$ to $(C_{48},c)$. This concludes the proof.
\end{proof}

We observe that the proof could be slightly modified to obtain variations of Theorem~\ref{thm:C48}.

\begin{remark}\label{rem:C48}\
\begin{enumerate}
\item In the reduction from Theorem~\ref{thm:C48}, Red vertices are never in the same part of the bipartition as Blue and Green vertices. It follows that one could colour every Red vertex Blue, and Theorem \ref{thm:C48} would still hold, for $3$-tropical cycles.
\item The idea of this proof can also be extended for a $2$-tropical $54$-cycle. To do this we first insert a Red vertex between $x_5$ and $x_6$ in $P_{xy}$ and a Red vertex between $z_5$ and $z_6$ in $Q_{zt}$. We observe that the proof follows similarly. However, in this case, all blue vertices are in one part and all green vertices are on the other part of the bipartition. Thus, as in the previous claim, we can remove two colours now and use the natural bipartition to distinguish two sets of colours for each colour class.
\end{enumerate}
\end{remark}

\longpaper{
\subsection{Other families of minimal graphs}

Next, we show that for each of the minimal graphs $H$ from Table~\ref{table} (other than even cycles) that make \PBlistCOL{$H$} NP-complete, \PBtropCOL{$H$} is polynomial-time solvable.

\begin{theorem}\label{thm:table}f
For every graph $H$ belonging to one of the six families $\mathcal T$, $\mathcal W$, $\mathcal D$, $\mathcal M$, $\mathcal N$ and $\mathcal G$ described in Table~\ref{table}, \PBtropCOL{$H$} is polynomial-time solvable.
\end{theorem}
\begin{proof}
We assume for contradiction, that for some integer~$i$ and a family $\mathcal F$ among $\mathcal T$, $\mathcal W$, $\mathcal D$, $\mathcal M$, $\mathcal N$ and $\mathcal G$, there is a problem \PBCOL{$(F_i,c)$} that is not polynomial-time solvable.

\medskip

\noindent\textbf{Family \boldmath{$\mathcal{T}$}.} Suppose $x_{i+4}$ is coloured~$m$. Suppose $y_{i+3}$ is coloured~$a$, $a\neq m$ by Proposition~\ref{prop:bipartite-hom}. Then, $y_{i+4}$ cannot be coloured~$a$ (otherwise it can be folded onto $y_{i+3}$), so it is coloured~$b$. Because of Lemma~\ref{lemm:unique-feature}, there must be another $P_3$ coloured $amb$ on the graph, but for the graph to be a core, the vertex coloured $m$ of this $P_3$ must not be adjacent to $y_{i+3}$. However, note that $y_{i+3}$ is adjacent to every vertex of $X$ except for $x_1$ and $x_{i+3}$, both of which have degree~$1$ and cannot create a $3$-vertex path coloured $a$-$m$-$b$.

\medskip

\noindent\textbf{Family \boldmath{$\mathcal{W}$}.} Now, we consider $W_i$. We try to find a colouring $c$ of $W_i$ such that \PBCOL{$(W_i,c)$} is not polynomial-time solvable. Suppose $y_{i+2}$ is coloured with colour $a$. Suppose $x_{i+3}$ is coloured $m$. $y_{i+4}$ cannot be coloured~$a$, otherwise it can be folded onto $y_{i+2}$, so we may assume it is coloured~$b$. $y_{i+3}$ cannot be coloured~$b$, for otherwise $y_{i+4}$ can be folded onto it, so it is coloured~$a$ or~$d$. Suppose first that it is coloured~$d$. By Lemma~\ref{lemm:unique-feature}, there is another vertex coloured~$a$, and the only one which could not be folded onto $y_{i+2}$ is $y_{i+1}$, so it must be coloured~$a$. Similarly, $y_{1}$ is coloured~$d$. By Lemma~\ref{lemm:unique-feature}, there is another edge besides $x_{i+3},y_{i+4}$ with endpoints coloured~$m$ and~$b$, but for the graph to be a core, the edge $x_{i+3}y_{i+4}$ must not be able to fold onto it. However, it is easily verified that this is impossible. So, we must assume that $y_{i+3}$ is coloured~$a$. There is no other vertex in $Y$ coloured~$a$, otherwise it can be folded onto $y_{i+2}$ or $y_{i+3}$. We can assume, without loss of generality, that a connected subgraph of the source graph, coloured only with~$m$ and~$b$, and with only vertices of colour~$a$ at distance~$1$, will be sent to $x_{i+3}$ and $y_{i+4}$. Knowing this, we can contract each such subgraph to a single vertex, coloured with a new colour~$\omega$, and similarly replace $x_{i+3}$ and $y_{i+4}$ by a single vertex coloured~$\omega$, adjacent to $y_{i+2}$ and $y_{i+3}$. There will be a homomorphism between the source graph and $(W_i,c)$ if and only if there is one after such transformation. However, the graph obtained after such transformation will not contain any induced subgraph from the table above, which yields a contradiction.

\medskip

\noindent\textbf{Family \boldmath{$\mathcal{D}$}.} Now, consider $D_i$ and a colouring $c$ such that \PBCOL{$(D_i,c)$} is not polynomial-time solvable. Suppose $x_{i+4}$ is coloured~$m$ and $y_{i+4}$ is coloured~$a$. By Lemma~\ref{lemm:unique-feature}, there is another vertex coloured~$a$. We may assume that $y_1$ is such a vertex because it is the only one that cannot be folded on $y_{i+4}$. Then, $x_1$ cannot have colour~$m$, for otherwise it can be folded onto $x_{i+4}$, so it is coloured~$l$. By Lemma~\ref{lemm:unique-feature}, there is another vertex in $X$ coloured~$l$, say $v$. $y_1x_1$ can be folded onto $y_{i+4}v$, which yields a contradiction.

\medskip

\noindent\textbf{Family \boldmath{$\mathcal{M}$}.} Now, consider $M_i$ and a colouring $c$ such that \PBCOL{$(M_i,c)$} is not polynomial-time solvable. Suppose $x_2$ is coloured~$m$ and $y_{i+2}$ is coloured~$a$. By Lemma~\ref{lemm:unique-feature}, there is another vertex in $X$ coloured~$m$. The only vertex which can be coloured~$m$ without being able to be folded onto $x_2$ is $x_1$. This is because $y_{i+2}$ is adjacent only to $x_1$ and $x_2$ is adjacent to every vertex in $Y$ except for $y_{i+2}$. So we may assume $x_1$ is coloured~$m$. By Lemma~\ref{lemm:unique-feature}, there is another vertex in $Y$ coloured~$a$, say $v$. $x_1y_{i+2}$ can be folded onto $x_2v$ since $x_2$ is adjacent to every vertex in $Y$ except $y_{i+2}$, which yields a contradiction.

\medskip

\noindent\textbf{Family \boldmath{$\mathcal{N}$}.} Now, consider $N_i$ and a colouring $c$ such that \PBCOL{$(N_i,c)$} is not polynomial-time solvable. Suppose $x_2$ is coloured $m$. For $3\leq j\leq 2i+3$, $x_j$ cannot be coloured~$m$, for otherwise it can be folded onto $x_2$ since $N(x_j) \subset N(X_2)$. By Lemma~\ref{lemm:unique-feature}, $x_1$ or $x_{2i+4}$ must be coloured~$m$. Both $x_1$ and $x_{2i+4}$ have a neighbour of degree~$1$ (namely, $y_{i+1}$ and $y_{2i+4}$, respectively), which are the two only vertices in $Y$ not adjacent to $x_2$. By Lemma~\ref{lemm:unique-feature}, neither $x_1y_{i+1}$ nor $x_{2i+4}y_{2i+4}$ can be an edge of unique colour. Either exactly one of them is coloured~$ma$ and a neighbour $v$ of $x_2$ is coloured~$a$, in which case the graph is not a core because the edge can be folded on $x_2v$ (since $N(x_1) \setminus \{y_{i+1}\}$ and $N(x_{2i+4}) \setminus \{y_{2i+4}\}$ are both subsets of $N(x_2)$), or both $x_1y_{i+1}$ and $x_{2i+4}y_{2i+4}$ are coloured $ma$ and the graph is not a core because $x_{2i+4}y_{2i+4}$ can be folded on $x_1y_{i+1}$ since $N(x_{2i+4}) \setminus \{y_{2i+4} \} \subset N(x_1)$, yielding a contradiction.

\medskip

\noindent\textbf{Family \boldmath{$\mathcal{G}$}.} We try to find a colouring $c$ of $G_1$ such that \PBCOL{$(G_1,c)$} is not polynomial-time solvable. The colour of $y_1$ is, say, $a$. By Lemma~\ref{lemm:unique-feature}, colour~$a$ must be present somewhere else in $Y$. By symmetry, we can assume $y_2$ is coloured~$a$. The two neighbours of $y_2$ cannot be coloured with the same colour, for otherwise we can fold $x_2$ on $x_1$, implying that $(G_1,c)$ is not a core, a contradiction. Without loss of generality, $x_1$ and $x_2$ are coloured~$1$ and~$2$ respectively. By Lemma~\ref{lemm:unique-feature} applied to edge $y_2x_2$, there must be another edge coloured~$a2$. However, if a neighbour of $y_1$ is coloured~$2$, we can fold $y_2x_2$ onto $y_1$ and the graph is not a core, a contradiction. It follows that the other edge coloured~$a2$ is either $y_3x_4$ or $y_4x_6$. By symmetry, we can assume that $x_4$ is coloured~$2$ and $y_3$ is coloured~$a$. $x_3$ cannot be coloured~$1$ or~$2$, for otherwise $(G_1,c)$ is not a core. Therefore, $x_3$ is coloured with a third colour, say~$3$. At this point, $y_1x_1y_2x_2$ is coloured $a1a2$ and $y_1x_3y_3x_4$ is coloured $a3a2$. Consider the colour of $y_4$. It must be $a$ by Lemma~\ref{lemm:unique-feature}. There are only two uncoloured vertices, $x_5$ and $x_6$, which must be coloured~$1$ and~$3$ by Lemma~\ref{lemm:unique-feature}. The graph is not a core in both cases as we can either fold $x_6y_4$ onto $x_1y_1$ or the edge $x_6y_4$ onto $x_3y_1$, a contradiction.

Now, let $c$ be a colouring of $G_2$ such that \PBCOL{$(G_2,c)$} is not polynomial-time solvable. Suppose the vertex $y_2$ is coloured with $a$. Then, $y_1$ cannot be coloured~$a$, for otherwise it can be folded onto $y_2$, which yields a contradiction. Therefore, $y_1$ is coloured~$b$. Because of Lemma~\ref{lemm:unique-feature}, $y_3$ and $y_4$ must be coloured~$a$ and~$b$. By symmetry, we may assume $y_3$ is coloured~$a$ and $y_4$ is coloured~$b$. Suppose $x_5$ is coloured~$m$. Then $x_1$, $x_2$, $x_3$ and $x_4$ cannot be coloured~$m$, for otherwise $y_3x_5$ can be folded on $y_2$. Thus, $y_3x_5$ is the only edge coloured~$am$. Lemma~\ref{lemm:unique-feature} yields a contradiction.

Now, let $c$ be a colouring of $G_3$ such that \PBCOL{$(G_3,c)$} is not polynomial-time solvable. By Lemma~\ref{lemm:unique-feature}, there are at most two colours in each part of the bipartition. If $x_1$ and $x_2$ have the same colour, $x_2$ can be folded onto $x_1$, a contradiction. Similarly, if $y_1$ and $y_4$ have the same colour, $y_1$ can be folded onto $y_4$. Then $x_1$, $y_1$, $x_2$ and $y_4$ induce a complete bipartite graph with every colour of $c$, implying that $(G_3,c)$ is not a core, a contradiction.
\end{proof}

}

\section{Bipartite graphs of small order}\label{sec:smallgraphs}

In this section, we show that for each graph $H$ of order at most~$8$, \PBtropCOL{$H$} is polynomial-time solvable. On the other hand, there is a graph $H_9$ of order~$9$ such that \PBtropCOL{$H_9$} is NP-complete.

\begin{theorem} \label{thm:small_graph}
For any bipartite graph $H$ of order at most~$8$, \PBtropCOL{$H$} is polynomial-time solvable.
\end{theorem}
\begin{proof}
It suffices to prove that for each bipartite graph $H$ of order at most~$8$ and each colouring $c$ of $H$, \PBCOL{$(H,c)$} is polynomial-time solvable. In fact, by Proposition~\ref{prop:bipartite-hom} it suffices to show the statement for colourings of $H$ such that the colour sets in the two parts of the bipartition are disjoint. To prove that \PBCOL{$(H,c)$} is polynomial-time solvable it is enough to prove it for the core of $S(H,c)$, it is also enough to prove it for each connected component of $(H,c)$. Thus in the rest of the proof we always assume that $(H,c)$ is connected core. Let $(X,Y)$ be the bipartition of $H$. 

Since the only graphs of order at most~$8$ in the characterization of minimal NP-complete graphs $H$ with \PBlistCOL{$H$} NP-complete are the cycles $C_6$ and $C_8$~\cite{FHH99} \longpaper{(see Table~\ref{table})}, by Theorem~\ref{thm:listhom-table}, if $H$ does not contain an induced $6$-cycle or an induced $8$-cycle, then \PBlistCOL{$H$} is polynomial-time solvable and therefore \PBtropCOL{$H$} is polynomial-time solvable. Therefore $H$ contains an induced $6$-cycle or an induced $8$-cycle.

If $H$ contains an induced copy of $C_8$, then $H$ is isomorphic to $C_8$ itself and hence we are done by Theorem~\ref{thm: C12}. Therefore, we can assume that $H$ contains an induced copy of $C_6$. Again by Theorem~\ref{thm: C12}, if $H$ is isomorphic to $C_6$, we are done.

Now, assume that $H$ is a bipartite graph of order~$7$ or~$8$ with an induced copy of $C_6$. If one part, say $X$, is of order~$3$, then all its vertices belong to each $6$-cycle of $H$. Hence, for each $x\in X$, \PBlistCOL{($H-x$)} is polynomial-time solvable. Thus, if $X$ is not monochromatic, we can apply Lemma~\ref{lemm:unique-feature} and \PBCOL{$(H,c)$} is polynomial-time solvable. Therefore we may assume $X$ is monochromatic, say Blue. If $Y$ contains at most two colours, then $(H,c)$ contains as a subgraph the path on three vertices where the central vertex is Blue and the other vertices are coloured with the colours of $Y$.But then $(H,c)$ maps to this subgraph and, therefore, it is not a core, a contradiction. Hence, $Y$ contains at least three colours. If $|Y|=4$, then $Y$ contains two colours that are the unique ones coloured with their colour. Moreover, \PBlistCOL{$(H-\{x,y\})$} contains no $6$-cycle and, therefore, by Lemma~\ref{lemm:unique-feature} \PBCOL{$(H,c)$}, is polynomial-time solvable. Hence we can assume that $|Y|=5$. If $Y$ contains at least four colours, by the same argument we are done, therefore, we assume that $Y$ contains exactly three colours. If $(H,c)$ contains a star with a Blue centre and a three leaves of different colours, then $(H,c)$ is not a core. Therefore the neighbourhood of each vertex of $X$ contains at most two colours.
Assume that the three vertices $y_1$, $y_2$, $y_3$ of $Y$ in the $6$-cycle have three different colours. Let $y_4$, another element of $Y$ be of the same colour as $y_i$. By the previous observation, $y_4$ can only be adjacent to neighbours of $y_i$. But then mapping $y_4$ to $y_1$ is a homomorphism which means $(H,c)$ is not a core.
Therefore, we can assume that $c(y_1)=c(y_2)=1$ and $c(y_3)=2$. Then, the vertex coloured~$3$ has degree~$1$ and is adjacent to the common neighbour of $y_1$ and $y_2$. But then again, $(H,c)$ is not a core.

Therefore, $H$ is a bipartite graph of order~$8$ and $|X|=|Y|=4$. If there are at least three colours in one part of the bipartition (say $X$), then two vertices $x_1$, $x_2$ in $X$ form two colour classes of size~$1$. Moreover, $H-\{x_1,x_2\}$ has no $6$-cycle and therefore, by Lemma~\ref{lemm:unique-feature}, \PBCOL{$(H,c)$} is polynomial-time solvable. We may then assume that each part of the bipartition contains at most two colours. If one part, say $X$, contains exactly one colour (say Blue), then $(H,c)$ contains a path on three vertices with every colour of $c$ (the central vertex is Blue) and is not a core, a contradiction. Therefore each part of the bipartition contains exactly two colours. If in each part, each colour has exactly two vertices, we can apply Lemma~\ref{lemm:2SAT} to show that \PBCOL{$(H,c)$} is polynomial-time solvable. Therefore, we can assume that there is a colour, say Blue, where exactly three vertices of one part, say $x_1$, $x_2$, $x_3$ from part $X$, coloured Blue ($x_4$ is coloured Green). If $H-x_4$ contains no induced $6$-cycle (it cannot contain an $8$-cycle since it has order~$7$), then \PBlistCOL{$(H-x_4)$} is polynomial-time solvable and we can use Lemma~\ref{lemm:unique-feature} and \PBCOL{$(H,c)$} is polynomial-time solvable. Hence we may assume $H-x_4$ contains an induced $6$-cycle $C$. Note that $C$ must contain three vertices of $X$ and therefore contains all three of $x_1$, $x_2$, $x_3$. If the three other vertices $y_1$, $y_2$ an $y_3$ of $C$ are coloured with the same colour, then $(H,c)$ is not a core, a contradiction. Therefore assume, without loss of generality, that $c(y_1)=c(y_2)=1$ and $c(y_3)=2$. Then, in order for $(H,c)$ to be a core, we cannot have both $x_1$ and $y_1$ (respectively, $y_2$ and $x_3$) of degree~$3$. More precisely, either $d(y_1)=d(x_3)=2$ and $d(x_1)=d(y_2)=3$, or $d(y_1)=d(x_3)=3$ and $d(x_1)=d(y_2)=2$. In both cases, we have $d(y_3)=2$, for otherwise $(H,c)$ contains a $4$-cycle with all four colours, and $(H,c)$ is not a core. If $c(y_4)=1$, then $(H,c)$ contains a path on four vertices coloured $2$-Blue-$1$-Green; moreover there is no edge in $(H,c)$ whose endpoints are coloured Green and $2$, therefore $(H,c)$ is homomorphic to the above path and is not a core. If $c(y_4)=2$, then $(H,c)$ contains a $4$-coloured $4$-cycle and again $(H,c)$ is not a core, a contradiction. As no such tropical graph exists, we have shown that for all possible cases the $(H,c)$-colouring problem is polynomial-time solvable.
\end{proof}

Denote by $H_9$ the graph obtained from a $6$-cycle by adding a pendant degree~$1$-vertex to three independent vertices (see Figure~\ref{fig:G9}).

\begin{figure}[ht!]
\centering
\includegraphics[scale=0.7]{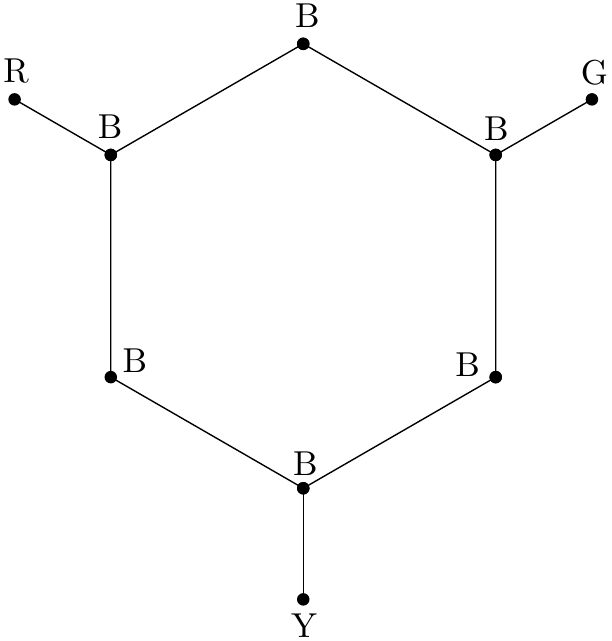}
\caption{The $4$-tropical graph $H_9$.}
\label{fig:G9}
\end{figure}

\begin{theorem}\label{thm:H_9}
\PBtropCOL{$H_9$} is NP-complete.
\end{theorem}
\begin{proof}
We show that \PBCOL{$(H_9,c)$} is NP-complete, where $c$ is the $4$-colouring of $H_9$ illustrated in Figure~\ref{fig:G9}. We describe a reduction from \PBlistCOL{$C_6$}, which is NP-complete~\cite{FHH99}. We label the vertices in $C_6$ from~$1$ to~$6$ sequentially. We also do that in the $C_6$ included in $H_9$. We assume without loss of generality that the vertex adjacent to the Red vertex is labelled~$1$, and the one adjacent to the Green one is labelled~$3$. It follows that the vertex adjacent to the Yellow vertex is labelled~$5$.

Let $(G,L)$ be an instance of \PBlistCOL{$C_6$}, where $L$ is the list-assignment function. If $G$ is not bipartite, then $G$ has no homomorphism to $C_6$, so we can assume that $G$ is bipartite. Since $G$ and $C_6$ are bipartite, we may assume that $\forall u \in V(G)$, either $L(u)\subseteq\{1,3,5\}$, or $L(u)\subseteq\{2,4,6\}$. Thus $|L(u)|\leq 3$.

From $(G,L)$, we build an instance $f(G,L)$ of \PBCOL{$(H_9,c)$} as follows. First, we consider a copy $G'$ of $G$, we let $G'\subset f(G,L)$ and colour every vertex of $G'$ Black. We call $u'$ the copy of vertex $u$ in $G'$. Then, for each vertex $u$ of $G$, we add a gadget $H_u$ to $f(G,L)$ that is attached to $u'$. The gadget is described below and depends only on $L(u)$.

\begin{itemize}
\item If $L(u)=\{1\}$ (respectively, $\{3\}$ or $\{5\}$), then $H_u$ is a single Red (respectively, Green or Yellow) vertex of degree~$1$ adjacent only to $u'$.
\item If $L(u)=\{2\}$ (respectively, $\{4\}$ or $\{6\}$), then $H_u$ consists of two $2$-vertex path: a Red--Black path and a Green--Black path (respectively, a Green--Black path and a Yellow--Black path or a Yellow--Black path and a Red--Black path) whose Black vertex is of degree~$2$ and is adjacent to $u'$ (the other vertex is of degree~$1$).
\item If $L(u)=\{2,4\}$ (respectively, $\{4,6\}$ or $\{2,6\}$), then $H_u$ is a $2$-vertex Green--Black (respectively, Yellow--Black or Red--Black) path whose Black vertex is of degree~$2$ and adjacent to $u'$ (the other vertex is of degree~$1$).
\item If $L(u)=\{1,3\}$ (respectively, $\{3,5\}$ or $\{1,5\}$), then $H_u$ is a $5$-vertex Red--Black--Black--Black--Green path (respectively, Green--Black--Black--Black--Yellow or Yellow--Black--Black--Black--Red) whose middle Black vertex is of degree~$3$ and adjacent to $u'$ (the endpoints of the path are of degree~$1$ and the other two vertices have degree~$2$).
\item If $L(u)=\{1,3,5\}$, then $H_u$ is a $3$-vertex Black--Black--Red path with the black leaf adjacent to $u'$.
 \item If $L(u)=\{2,4,6\}$, then $H_u$ is a $4$-vertex Black--Black--Black--Red path with the black leaf adjacent to $u'$.
\end{itemize}

Let us prove that $G$ has a homomorphism to $C_6$ that fulfils the constraints of list $L$, if and only if $f(G,L)\to (H_9,c)$.

For the first direction, consider a list homomorphism $h$ of $G$ to $C_6$ with the list function $L$. We build a homomorphism $h'$ of $f(G,L)$ to $(H_9,c)$ as follows. First of all, each copy $v'$ of a vertex $v$ of $G$ with $h(v)=i$ is mapped to $i$ in $(H_9,c)$. It is clear that this defines a homomorphism of the subgraph $G'$ of $f(G,L)$ to the Black $6$-cycle in $(H_9,c)$. It is now easy to complete $h'$ into a homomorphism of $f(G,L)$ to $(H_9,c)$ by considering each gadget $H_u$ independently.

For the converse, let $h_T$ be a homomorphism of $f(G,L)$ to $(H_9,c)$. Then, we claim that the restriction of $h_T$ to the vertices of the subgraph $G'$ of $f(G,L)$ is a list homomorphism of $G$ to $C_6$ with list function $L$. Indeed, let $u'$ be a vertex of $G'$. If $H_u$ has one vertex (say a Red vertex), then $L(u)=\{1\}$. Then necessarily $u'$ is sent to a neighbour of a vertex coloured Red in $(H_9,c)$. Since the only such neighbour is vertex~$1$, $u'\in h_T(u)$. All the other cases follow from similar considerations.
\end{proof}

\section{Trees}\label{sec:trees}

We now consider the complexity of tropical homomorphism problems when the target tropical graph is a tropical tree.

It follows from the results in Section~\ref{sec:list_hom} that for every tree $T$ of order at most~$10$, \PBtropCOL{$T$} is polynomial-time solvable. Indeed, such a tree needs to contain a minimal tree $T$ of order at most~$10$ for which \PBlistCOL{$T$} is NP-complete, and the only such tree is $G_1$, which has order~$10$~\cite{FHH99}. \longpaper{(See Table~\ref{table}.)} We proved in Theorem~\ref{thm:table} that \PBtropCOL{$G_1$} is polynomial-time solvable. With some efforts, one can extend this to trees of order at most~$11$.
\shortpaper{The proof is tedious and we omit it here, see~\cite{fullversion} for details.}

\begin{theorem}\label{thm:SmallTree}
For every tree $T$ of order at most~$11$, \PBtropCOL{$T$} is polynomial-time solvable.
\end{theorem}
\longpaper{
\begin{proof}
Let $G_1$ be the smallest tree such that \PBlistCOL{$G_1$} is NP-hard, as defined in Table~\ref{table} of Section~\ref{sec:list_hom} ($G_1$ has order $10$ and is obtained from a claw by subdividing each edge twice). We let $V(G_1)=\{c,x_1,y_1,z_1,x_2,y_2,z_2,x_3,y_3,z_3\}$, with edges $cx_i$, $x_iy_i$, $y_iz_i$ for $i=1,2,3$.

Assume for a contradiction that there is a tree $T_0$ of order~$11$ such that \PBtropCOL{$T$} is not polynomial-time solvable. Then, $T_0$ is a connected core. Once again, by Proposition~\ref{prop:bipartite-hom}, we may assume that the colour sets of the two parts in the bipartition of $T_0$ are disjoint. By Theorem~\ref{thm:listhom-table}, for any tree $T$ which does not contain $G_1$ as an induced subgraph, \PBlistCOL{$T$} is polynomial-time solvable, and therefore \PBtropCOL{$T$} is polynomial-time solvable. Hence $G_1$ is a subtree of $T_0$.

There are four non-isomorphic trees of order~$11$ which contain $G_1$, depending on where we attach the additional vertex $a$. If in $T_0$, $a$ is adjacent to $c$, then the same arguments as in the proof of Theorem~\ref{thm:table} showing that \PBtropCOL{$G_1$} is polynomial-time solvable show that \PBtropCOL{$T_0$} is polynomial-time solvable, a contradiction.

Let $(A,B)$ be the bipartition of $T_0$ with $\{c,y_1,y_2,y_3\}\subseteq A$ and $\{x_1,x_2,x_3,z_1,z_2,z_3\}\subseteq B$. For the remainder, we may assume that no vertex (except $a$) is the only one with its colour, for otherwise, by Lemma~\ref{lemm:unique-feature}, \PBtropCOL{$T_0$} would be polynomial-time solvable. In particular, $A-a$ is coloured with at most two colours and $B-a$ is coloured with at most three colours.

Assume first that $a$ is adjacent to a vertex $x_i$ of $G_1$, say $x_1$. The colours of $x_1$ and $z_1$ must be distinct, otherwise $(T_0,c_0)$ is not a core. Without loss of generality, assume that $c_0(x_1)=1$ and $c_0(z_1)=2$. Without loss of generality the central vertex $c$ is Black. The supplementary vertex $a$ must be coloured with a different colour than $c$ and $y_1$ (say with colour Red), otherwise $(T_0,c_0)$ is not a core. Hence $y_1$ is not Red. Assume first that $y_1$ is Green. Then (without loss of generality), $y_2$ is Black and $y_3$ is Green, otherwise we could apply Lemma~\ref{lemm:unique-feature}. But by Lemma~\ref{lemm:unique-feature}, there must be two edges with endpoints $1$ and Green, and one with endpoints $2$ and Green. Hence $c_0(x_3)=1$ and $c_0(z_3)=2$ (if $c_0(x_3)=2$ and $c_0(z_3)=1$ then $(T_0,c_0)$ is not a core). But again by Lemma~\ref{lemm:unique-feature} we need another edge with endpoints Black and $1$, and one with endpoints Black and $2$. But in both cases $(T_0,c_0)$ is not a core, a contradiction. This shows that vertex $y_1$ must be Black. Then, since $(T_0,c_0)$ is a core, vertex $c$ has no neighbour coloured~$2$. But if there is no second edge with endpoints coloured $2$ and Black, then we could apply Lemma~\ref{lemm:unique-feature}. Hence one of $y_2$ and $y_3$, say $y_2$, must be Black, and $c_0(z_2)=2$. If $c_0(x_2)=1$, $(T_0,c_0)$ is not a core, therefore $c_0(x_2)=3$, and $c_0(x_3)\in\{1,3\}$. If $y_3$ is Black, then $(T_0,c_0)$ is not a core, hence $y_3$ is Red. But both neighbours of $y_3$ must have distinct colours, which means we can apply Lemma~\ref{lemm:unique-feature} to one of the edges incident with $y_3$, a contradiction.

Assume now that $a$ is adjacent to a vertex $y_i$ of $G_1$, say $y_1$. Then, the colours of $a$, $x_1$ and $z_1$ must be distinct, say $c_0(x_1)=1$, $c_0(z_1)=2$, $c_0(a)=3$. Without loss of generality the central vertex $c$ is Black. By Lemma~\ref{lemm:unique-feature}, there is another vertex coloured Black. If $y_1$ is Black, then by Lemma~\ref{lemm:unique-feature} we have two further edges with endpoints Black-$2$ and Black-$3$. But these edges cannot be both incident with $c$ (otherwise $(T_0,c_0)$ is not a core), hence there is another Black vertex. Then in fact, Lemma~\ref{lemm:unique-feature} implies that both $y_2$ and $y_3$ are Black. But then, any way to complete $c_0$ implies that $(T_0,c_0)$ is not a core, a contradiction. Therefore, $y_1$ is not Black (say it is Red) and we can assume that $y_2$ is Black, and since we need a second Red vertex, $y_3$ is Red. But one of the type of edges among Red-$1$, Red-$2$ and Red-$3$ will appear only once, and we can apply Lemma~\ref{lemm:unique-feature}, a contradiction.

We assume finally that $a$ is adjacent to a vertex $z_i$ of $G_1$, say $z_1$. Without loss of generality, vertex $a$ is Black, vertex $z_1$ is coloured~$1$, and vertex $y_1$ is Red (otherwise, $(T_0,c_0)$ is not a core). By Lemma~\ref{lemm:unique-feature}, there must be another $3$-vertex path coloured Black-$1$-Red. This path must be $cx_iy_i$ with $c$ Black, for otherwise $(T_0,c_0)$ is not a core. We can assume that $c_0(x_2)=1$ and $y_1$ is Red. Then $c_0(x_1)\neq 1$, assume $c_0(x_1)=2$. Then again by Lemma~\ref{lemm:unique-feature} there is another $3$-vertex path coloured Black-$2$-Red. The only possibility is that $c_0(x_3)=2$ and $y_3$ is Red. Then $c_0(z_3)\notin\{1,2\}$, otherwise $(T_0,c_0)$ is not a core. Hence we assume $c_0(z_3)=3$, which by Lemma~\ref{lemm:unique-feature} implies $c_0(z_2)=3$. But then there is a unique $3$-vertex path coloured $1$-Red-$3$, and by Lemma~\ref{lemm:unique-feature}, \PBCOL{$(T_0,c_0)$} is polynomial-time solvable, a contradiction. This completes the proof.
\end{proof}
}


Let $T_{23}$ be the tree of order~$23$ shown in Figure~\ref{target_23_tree}.

\begin{figure}[ht!]
\centering
\includegraphics[scale=0.7]{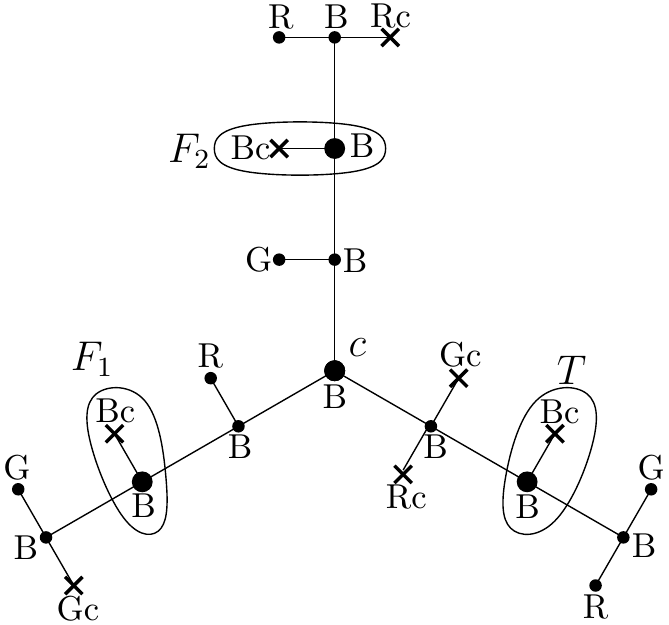}
\caption{The $7$-tropical tree $(T_{23},c)$}
\label{target_23_tree}
\end{figure}

\begin{theorem}\label{thm:T23}
\PBtropCOL{$T_{23}$} is NP-complete.
\end{theorem}
\begin{proof}
We give a reduction from \textsc{$3$-SAT} to \PBCOL{$(T_{23},c)$}, where $c$ is the colouring of Figure~\ref{target_23_tree}. Given an instance $(X,C)$ of \textsc{$3$-SAT}, we construct an instance $f(X,C)=(G_{X,C},c_{X,C})$ of \PBCOL{$(T_{23},c)$}.

To construct the graph $G_{X,C}$, we first define the following building blocks. See Figure~\ref{building_block_tree} for illustrations.

\begin{itemize}
\item The block $S_{1,2}$ is a graph built from a $7$-vertex black-coloured path with vertex set $\{x_1,\ldots,x_7\}$ where a BlackCross leaf is attached to vertices $x_1$ and $x_7$, a RedDot leaf is attached to vertices $x_2$ and $x_6$, and a GreenDot leaf is attached to vertex $x_4$.
\item The block $S_{1,T}$ is a graph built from a $7$-vertex black-coloured path with vertex set $\{x_1,\ldots,x_7\}$ where a BlackCross leaf is attached to vertices $x_1$ and $x_7$, a RedDot leaf is attached to vertices $x_2$ and $x_6$, and a RedCross leaf is attached to vertex $x_4$.
\item The block $S_{1,T}$ is a graph built from a $7$-vertex black-coloured path with vertex set $\{x_1,\ldots,x_7\}$ where a BlackCross leaf is attached to vertices $x_1$ and $x_7$, a GreenDot leaf is attached to vertices $x_2$ and $x_6$, and a GreenCross leaf is attached to vertex $x_4$.
\item The \emph{NOT-block} is depicted in Figure~\ref{fig:NOT-block}.
\item The \emph{A-block} is depicted in Figure~\ref{fig:A-block}.
\end{itemize}

Illustrations of these blocks can be found in Figure~\ref{building_block_tree}.

\begin{figure}[ht!]
\centering
\subfigure[The blocks $S_{1,2}$, $S_{1,T}$ and $S_{2,T}$.]{\label{fig:S-blocks}\includegraphics[scale=0.7]{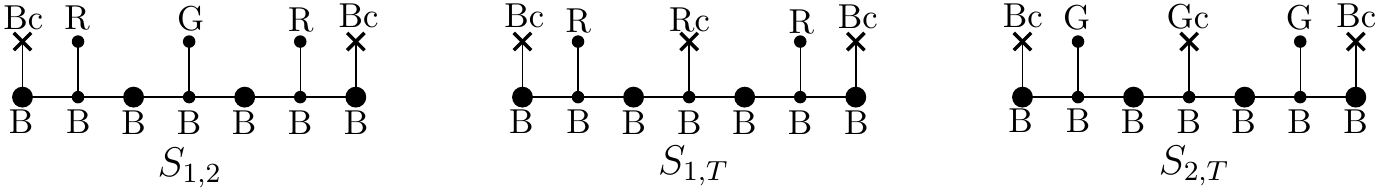}}\qquad
\subfigure[The variable gadget of $x$, essentially a NOT-block.]{\label{fig:NOT-block}\includegraphics[scale=0.7]{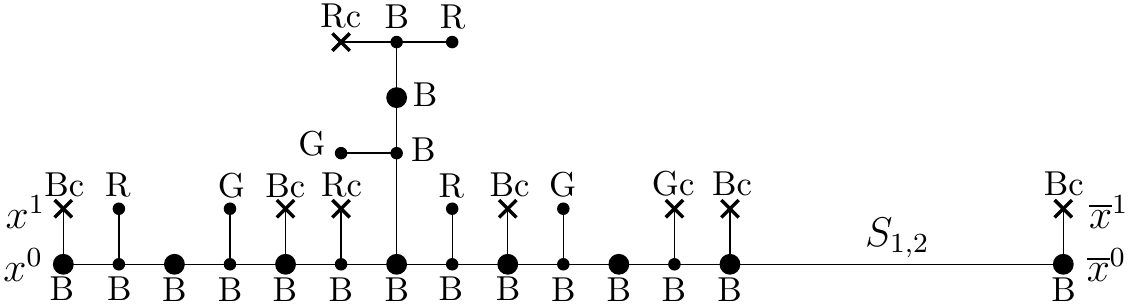}}\qquad
\subfigure[The $A$-block and its representation as an arrow.]{\label{fig:A-block}\includegraphics[scale=0.7]{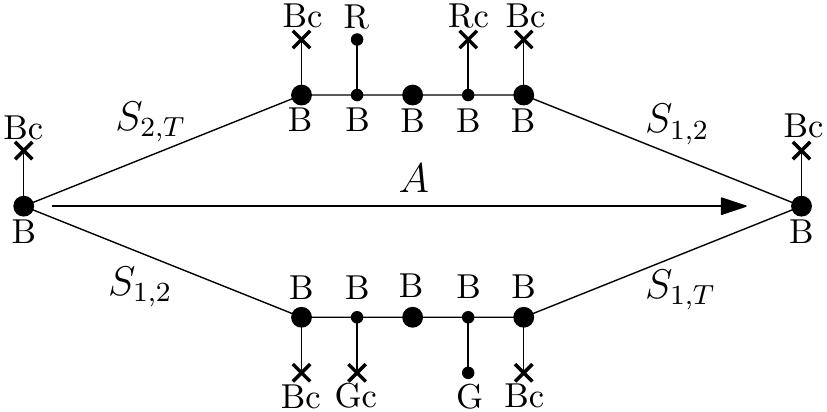}}
\caption{The building blocks of $G_{X,C}$.}
\label{building_block_tree}
\end{figure}

We now define gadgets for each variable of $X$ and each clause of $C$. The graph $G_{X,C}$ is formed by the set of all variable and clause gadgets.

\begin{itemize}

\item  For a variable $x\in X$, the \emph{variable gadget of $x$} consists of the four vertices $x^0$, $x^1$, $\bar{x}^0$ and $\bar{x}^1$, coloured respectively BlackDot, BlackCross, BlackDot and BlackCross, joined by a NOT-block as described in Figure \ref{fig:NOT-block}. The image of $x^0$ and $x^1$ in $(T_{23},c)$ correspond to the truth-value of the litteral $x$. Similarly, the image of $\bar{x}^0$ and $\bar{x}^1$ correspond to the truth-value of the litteral $\bar{x}$. For a litteral $l$, we use the notation $l^0$ (resp. $l^1$) to describe either $x^0$ (resp. $x^1)$ when $l=x$ with $x\in X$, or $\bar{x}^0$ (resp. $\bar{x}^1$) when $l=\bar{x}$ with $x\in X$.

\item For each clause $c=(l_1,l_2,l_3)\in C$, there is a \emph{clause gadget of $c$} (as drawn in Figure~\ref{gadget_clause_tree}) connecting vertices $l_1^0$, $l_2^0$ and $l_3^0$.
\end{itemize}

\begin{figure}[ht!]
\centering
\includegraphics[scale=0.7]{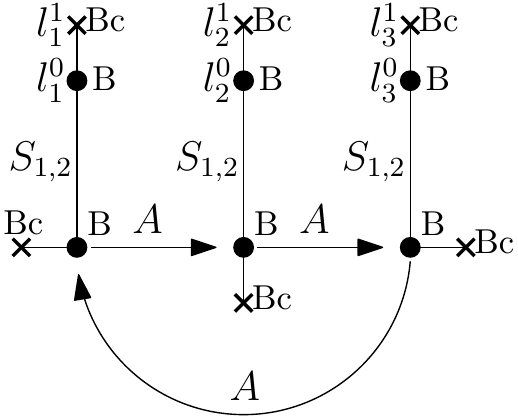}
\caption{Example of a clause gadget of clause $(l_1,l_2,l_3)$. The full details of the $A$-blocks and $S_{1,2}$-blocks are represented in Figure~\ref{building_block_tree}.}
\label{gadget_clause_tree}
\end{figure}

We now show that $G_{X,C}\to (T_{23},c)$ if and only if $(X,C)$ is satisfiable.

\medskip

Assume first that there is a homomorphism $h$ of $G_{X,C}$ to $(T_{23},c)$. We first prove some properties of $h$.

\begin{claim}\label{clm:T_{23}}
 The homomorphism $h$ satisfies the following properties.
\begin{itemize}
\item[(1)] For each literal $l$ of a variable of $X$, vertices $l^0$ and $l^1$ are mapped to the two vertices of one of the pairs $T$, $F_1$ or $F_2$. The same holds for the extremities of the blocks $S_{1,2}$, $S_{1,T}$, $S_{2,T}$ and $A$.
\item[(2)] The two extremities of each block $S_{1,2}$ are both mapped either to the vertices of $T$, or to vertices of $F_1\cup F_2$.
\item[(3)] The two extremities of each block $S_{1,T}$ are both mapped either to the vertices of $F_2$, or to vertices of $F_1\cup T$.
\item[(4)] The two extremities of each block $S_{2,T}$ are both mapped either to the vertices of $F_1$, or to vertices of $F_2\cup T$.
\item[(5)] For each variable $x$ of $X$, exactly one of $x^0$ and $\bar{x}^0$ is mapped to a vertex of $T$, and the other is mapped to a vertex of $F_1$ or $F_2$.
\item[(6)] In any $A$-block, either some extremity is mapped to $T$ (then the other extremity can be mapped to any of $F_1$, $F_2$ or $T$), or the left extremity is mapped to $F_2$ and the right extremity, to $F_1$.
\end{itemize}
\end{claim}
\noindent\emph{Proof of claim.}
(1) This is immediate since the only pairs in $(T_{23},c)$ consisting of two adjacent BlackDot and BlackCross vertices are the ones of $T$, $F_1$ and $F_2$.

\smallskip

(2)--(4) We only prove~(2), since the three proofs are not difficult and similar. By~(1), the extremities of $S_{1,2}$ are mapped to vertices of $T\cup F_1\cup F_2$. If one extremity is mapped to $T$, the remainder of the mapping is forced and the claim follows. If one extremity is mapped to $F_1\cup F_2$, one can easily complete it to a mapping where the other extremity is mapped to either $F_1$ or $F_2$.

\smallskip

(5) By (1), $x^0$ and $\bar{x}^0$ must be mapped to a vertex of $T\cup F_1\cup F_2$. Without loss of generality, we  can assume that $x^0$ corresponds to the left extremity of the NOT-block $N_x$ connecting $x^0$ and $\bar{x}^0$. First assume that $x^0$ and $\bar{x}^0$ are mapped to the vertex of $T$ coloured BlackDot. Then, considering the vertices of $N_x$ from left to right, the mapping is forced and the degree~$3$-vertex of $N_x$ at distance~$2$ both of a RedDot and a RedCross vertex must be mapped to the vertex $c$ of $T_{23}$. But then, continuing towards the right of $N_x$, $\bar{x}^0$ cannot be mapped to a vertex of $T$. Therefore, we may assume that both $x^0$ and $\bar{x}^0$ are mapped to the BlackCross vertices of $F_1\cup F_2$. If $x^0$ is mapped to the BlackCross vertex in $F_1$, then again going through $N_x$ from left to right the mapping is forced; the central vertex of $N_x$ must be mapped to a vertex of $F_2$, and $\bar{x}^0$ must be mapped to a vertex of $T$, a contradiction. The same applied when $x^0$ is mapped to the BlackCross vertex in $F_2$, completing the proof of~(5).

\smallskip

(6) An $A$-block is composed of two parts: the upper part and the lower part. Observe that if the left extremity of an $A$-block is mapped to $F_1$, then using~(2) and~(4), the mapping of the upper part of the $A$-block is forced and the right extremity has to be mapped to $T$. Similarly, if the left extremity is mapped to $F_2$, by~(2) and~(3) the right extremity cannot be mapped to $F_2$. On the other hand, for all other combinations of mapping the extremities to $T$, $F_1$ or $F_2$ the mapping can be extended.~\smallqed

\medskip

We are ready to show how to construct the truth assignment $A(h)$. If $h(l^0)\in T$ for some literal $l$, we let $l$ be True and if $h(l^0)\in F_1\cup F_2$, we let $l$ be False. By Claim~\ref{clm:T_{23}}(5), this is a consistent truth assignment for $X$. For any clause $c=(l_1,l_2,l_3)$, in the clause gadget of $c$, we have three $A$-blocks forming a directed triangle. Hence, by Claim~\ref{clm:T_{23}}(6), there must be one of the three extremities of this triangle mapped to a vertex of $T$. Therefore, by Claim~\ref{clm:T_{23}}(2), at least one of the vertices $l_1^0$, $l_2^0$ and $l_3^0$ is mapped to $T$. This shows that $A(h)$ satisfies the formula $(X,C)$.

\medskip

Reciprocally, if there is a solution for $(X,C)$, one can build a homomorphism of $G_{X,C}$ to $(T_{23},c)$ by mapping, for each literal $l$, the vertices $l_0$ and $l_1$ to one of the vertex pairs $F_1$, $F_2$ and $T$ of $(T_{23},c)$ corresponding to the truth value of $l$ (if $l$ is False, we may choose one of $F_1$ and $F_2$ arbitrarily). Then, using Claim~\ref{clm:T_{23}}, one can easily complete this to a valid mapping.
\end{proof}

\section{Conclusion}

We have shown that the class of \PBCOL{$(H,c)$} problems has a very rich structure, since they fall into the classes of CSPs for which a dichotomy theorem would imply the truth of the Feder--Vardi Dichotomy Conjecture. Hence, we turned our attention to the class of \PBtropCOL{$H$} problems, for which a dichotomy theorem might exist. Despite some initial results in this direction, we have not been able to exhibit such a dichotomy, and leave this as the major open problem in this paper.

Towards a solution to this problem, we propose a simpler question. All bipartite graphs $H$ that we know with problem \PBtropCOL{$H$} being NP-complete contain, as an induced subgraph, either an even cycle of length at least~$6$ (for example cycles themselves or $H_9$), or the graph $G_1$\longpaper{ from Table~\ref{table}}, that is, a claw with each edge subdivided twice (this is the case for $T_{23}$). Hence, we ask the following. (A bipartite graph is \emph{chordal} if it contains no induced cycle of length at least~$6$.)

\begin{question} \label{conj:nocycle}
Is it true that for any chordal bipartite graph $H$ with no induced copy of $G_1$, \PBtropCOL{$H$} is polynomial-time solvable?
\end{question}

Note that Question~\ref{conj:nocycle} is not an attempt at giving an exact classification, since \PBtropCOL{$G_1$} and \PBtropCOL{$C_{2k}$} for $k\leq 6$ are polynomial-time solvable.

\smallskip

Another interesting question would be to consider the restriction of \PBtropCOL{$H$} to $2$-tropical graphs. Recall that by Remark~\ref{rem:C48}(2), one can slightly modify the gadgets from Theorem~\ref{thm:C48} and the colouring of the cycle, to obtain a $2$-colouring $c$ of $C_{54}$ such that \PBCOL{$(C_{54},c)$} is NP-complete.

\medskip

Finally, we relate our work to the \textsc{$(H,h,Y)$-Factoring} problem studied in~\cite{BM97} and mentioned in the introduction. Recall that \PBCOL{$(H,c)$} corresponds to \textsc{$(H,c,K_{|C|}^+)$-Factoring} where $K_{|C|}^+$ is the complete graph on $|C|$ vertices with all loops, and with $C$ the set of colours used by $c$. In~\cite{BM97}, the authors studied \textsc{$(H,h,Y)$-Factoring} when $Y$ has no loops. Using reductions from NP-complete \PBCOL{$D$} problems where $D$ is an oriented even cycle or an oriented tree, they proved that for any fixed graph $Y$ which is not a path on at most four vertices, there is an even cycle $C$ and a tree $T$ such that \textsc{$(C,h_C,Y)$-Factoring} and \textsc{$(T,h_T,Y)$-Factoring} are NP-complete (for some suitable homomorphisms $h_C$ and $h_T$). Note that $C$ and $T$ here are fairly large. We can strengthen these results as follows. Consider our reduction of Theorem~\ref{thm:C48} showing in particular that \PBtropCOL{$C_{48}$} is NP-complete. As noted in Remark~\ref{rem:C48}(1), the given colouring $c$ of $C_{48}$ can easily be made a proper colouring by separating the red vertices into two classes, according to which part of the bipartition of $C_{48}$ they belong to. Then, one can observe that $c$ is in fact a homomorphism to a tree $T_1$ obtained from a claw where one edge is subdivided once (the three vertices of degree~$1$ are coloured Blue, Black and Green, and the two other vertices are the two kinds of Red). Thus, for any graph $Y$ containing this subdivided claw as a subgraph, we deduce that \textsc{$(C_{48},c_{1|T_1},Y)$-Factoring} is NP-complete. We can use a similar approach for our result of Theorem~\ref{thm:T23}, that \PBtropCOL{$T_{23}$} is NP-complete. Note that the colouring $c_2$ we give is in fact a homomorphism to a tree $T_2$ which is obtained from a star with five branches by subdividing one edge once. Thus, for any graph $Y$ containing $T_2$ as a subgraph, \textsc{$(T_{23},c_{2|T_2},Y)$-Factoring} is NP-complete. Of course we can apply this argument by replacing $T_1$ and $T_2$ by the underlying graph of any loop-free homomorphic image of $(C_{48},c_1)$ and $(T_{23},c_2)$, respectively.

\vspace{0.5cm}
\noindent\textbf{Acknowledgements.} We thank Petru Valicov for initial discussions on the topic of this paper.

\end{document}